\documentclass[10pt,a4paper, abstracton,fleqn]{scrartcl}

\usepackage[latin1]{inputenc}
\usepackage[round,authoryear]{natbib}
\usepackage{amsthm,dcolumn,graphicx,multicol,fancyhdr}
\usepackage{latexsym,subfig,graphicx,ifthen}
\usepackage{rotating,amsfonts, color,amssymb, amsmath, amscd, array, enumerate,afterpage,dsfont}
\usepackage{capt-of}

\usepackage{hyperref}
\usepackage{pdflscape}
\usepackage{ stmaryrd }
\usepackage{multirow}
\usepackage{float}
\usepackage{mathabx}

\newtheorem{theorem}{Theorem}[section]
\newtheorem{lemma}{Lemma}[section]

\newtheorem{rem}{Remark}[section]

\renewenvironment{proof}[1][\negthickspace]{\textbf{Proof\thickspace#1.} }{\
\rule{0.5em}{0.5em}}

\makeatletter\@addtoreset{equation}{section}\makeatother
\makeatletter\@addtoreset{table}{section}\makeatother
\makeatletter\@addtoreset{figure}{section}\makeatother

\newcommand{\ckboldon}[1]{#1}

\newcommand{\ckbold}[1]{%
 \ifthenelse{\isundefined{\ckboldon}}{#1}{ \textbf{#1} }
}

\newcommand{\bs}{\boldsymbol}
\newcommand{\cred}{\color{black}}
\newcommand{\revtwo}{\color{black}}

\newcommand{\var}{\operatorname{var}}
\newcommand{\E}{\operatorname{E}}
\newcommand{\cov}{\operatorname{cov}}
\newcommand{\ls}{\leqslant}
\newcommand{\gs}{\geqslant}
\newcommand{\eps}{\varepsilon}

\newcommand{\dto}{\overset{\mathcal{D}}{\longrightarrow}}

\newcommand{\pto}{\overset{P}{\longrightarrow}}

\newcommand{\vth}{\vartheta}

\newcommand{\bS}{\boldsymbol{S}}

\newcommand{\ZZ}{{\mathbb Z}}

\newcommand{\be}{\begin{eqnarray}}
\newcommand{\ee}{\end{eqnarray}}
\newcommand{\bq}{\begin{eqnarray*}}
\newcommand{\eq}{\end{eqnarray*}}

\newcommand{\bD}{\boldsymbol{D}}

\begin{document}

\title{A novel change point approach for the detection of gas emission sources using remotely contained concentration data}
\author{Idris Eckley\footnote{Mathematics and Statistics, Lancaster University, Lancaster, LA1\,4YF, UK; \texttt{i.eckley@lancs.ac.uk} }
\and
 Claudia Kirch\footnote{ Otto-von-Guericke University Magdeburg, Department of Mathematics, Institute of Mathematical Stochastics,   Postfach 4120, 39106 Magdeburg, Germany;
\texttt{claudia.kirch@ovgu.de}}
$ $\footnote{Center for Behavioral Brain Sciences (CBBS), Magdeburg, Germany}
\and
Silke Weber \footnote{ Karlsruhe Institute of Technology (KIT), Institute of Stochastics, Englerstr. 2,76131 Karlsruhe, Germany;
\texttt{silke.weber@kit.edu}}
}

\date{\today}
\maketitle

\begin{abstract}
Motivated by an example from remote sensing of gas emission sources, we derive two novel  change point procedures for multivariate time series where, in contrast to classical change point literature, the changes are not required to be aligned in the different components of the time series. Instead  the change points are described by a functional relationship where the precise shape depends on unknown parameters of interest such as the source of the gas emission in the above example. Two different types of tests and the corresponding estimators for the unknown parameters describing the change locations are proposed. We derive the null asymptotics for both tests under weak assumptions on the error time series and show asymptotic consistency under alternatives. Furthermore, we prove consistency for the corresponding estimators of the parameters of interest. The small sample behavior of the methodology is assessed by means of a simulation study and the above remote sensing example analyzed in detail.
	\noindent
\end{abstract}

 \vspace{5mm}
  \vspace{5mm}


 \section{Introduction}\label{sec:intro}
Change point analysis has a long and rich tradition, dating back to the work of \cite{Page1954} and \cite{Hinkley1970}.  During the last decade change point methods have attracted considerable interest, 
leading to substantial development of both methodology and diverse areas of applications. Recent surveys are given by \cite{horvath2014extensions} as well as \cite{Truong2018}, whilst \cite{Killick2012} provide a valuable resource collating recent published and software contributions.
. 

These methods are of fundamental importance in many areas, 
including econometrics \citep{aue2012sequential,hlavka2017fourier}, medicine \citep{fried2004online}, neuroscience \citep{Aston2012}, ocean-engineering \citep{NamAstonEckleyKillick2015} 
and bioinformatics \citep{Rigaill2012}. 


The challenge of detecting changes in multivariate time series has recently received growing attention. Notable contributions include \cite{Aue2009, Matteson2014, Zhang2010}. Initial research on this important problem has focussed on approaches for detecting those times at which changes occur in all series, e.g. \cite{Aue2009, Siegmund2011, Zhang2010}. More recently, several contributions have sought to relax this rather restrictive assumption, see \cite{Preuss2015} or \cite{Bardwell2018} for example. 

This article considers a different, albeit multivariate, change point setting. Specifically, the work that we describe is inspired by an application arising from remote sensing, where the changes in each component of the multivariate time series are functionally related to the changes in other series.   Remote sensing of gas emissions has been of considerable interest to researchers for a number of years. Applications range from monitoring green house gases \citep{Chen2006}, toxic gas emissions \citep{Bhattacharjee2008} and monitoring emissions from carbon storage resources \citep{Hirst2017}. In many of these examples, the primary objective is to be able to successfully locate sources of emission and quantify the emission rate(s). 

\begin{figure}[b]
\centering
\includegraphics[width=3.5in]{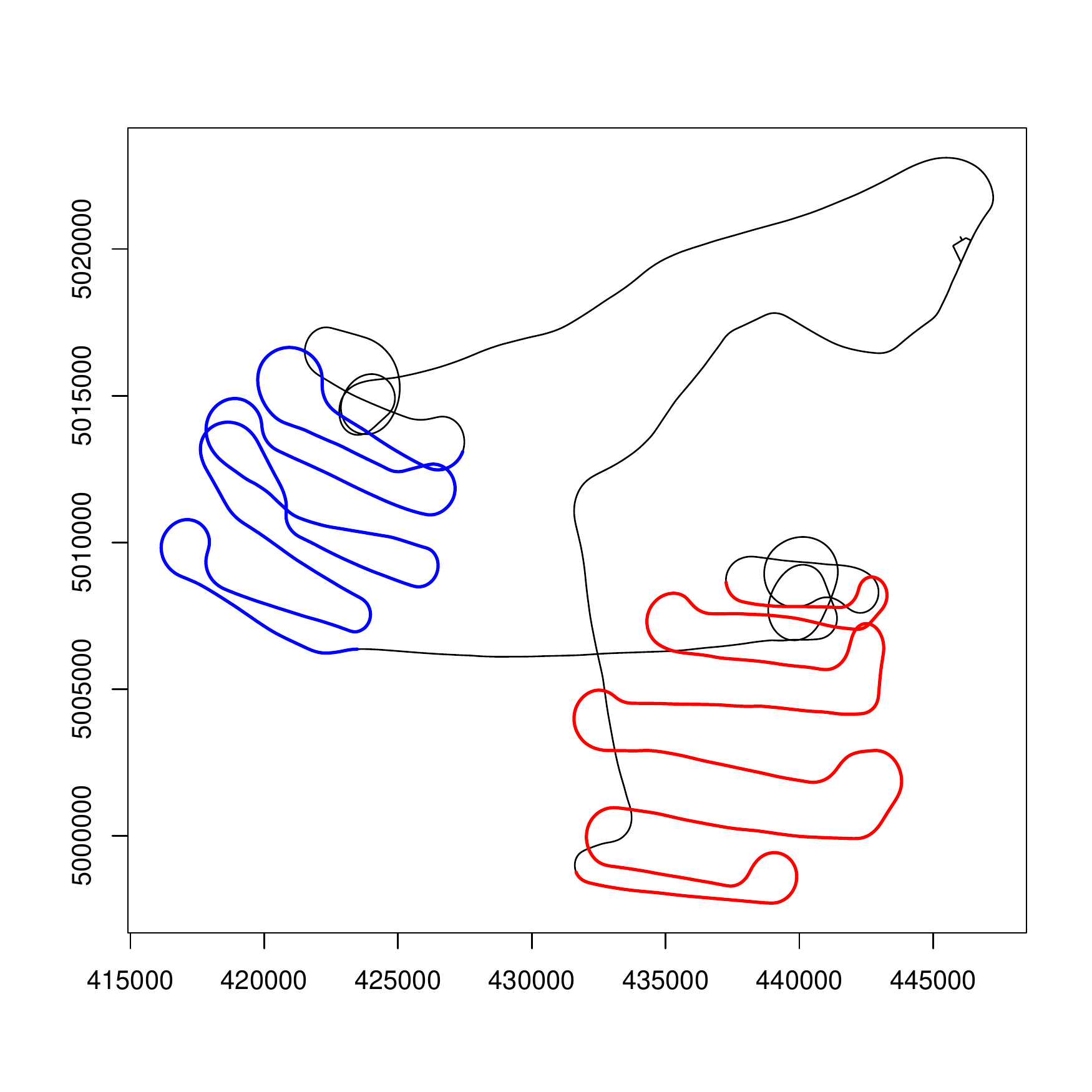}
\caption{Trajectory}
\label{trajectory}
\end{figure}

The application we consider centres on the remote detection and location of the source of gas emissions based on aerial sensed-data, as introduced in \cite{Hirst2013}. Their approach consists of an ultra-sensitive, high precision methane gas sensor mounted on an aircraft to measure a continuous stream of air from the leading edge of a wing. The sensor samples data at a high rate with GPS, radar altitude, barometric pressure, air temperature, wind velocity and several other variables. Flight data are then combined with meteorological data, including additional physical modelling attributes, including wind direction and atmospheric boundary layer depth, to estimate the shape of the plume and, thereby locate the source of the emission origin. 

%

The data that we consider, made available to us by \cite{Hirst2013}, provides a valuable test resource with known source locations. It is based on the atmospheric methane concentrations in the vicinity of two landfill sites. Specifically the data are collected by an aircraft flying at approximately 200m above ground level at a constant speed. This is well below the atmospheric boundary layer, that can constitute a `ceiling' on gases being transported from the ground. The aircraft surveys an area of approximately 40km $\times$ 40km, tracing back and forth in a snake-like fashion downwind of each landfill. Initial average wind speed and direction are also provided at multiple altitudes, see \cite{Hirst2013} for details. Figure \ref{trajectory} shows the flight trajectory in the vicinity of the landfill sites. 

\begin{figure}[b]
\centering
\includegraphics[width=5in]{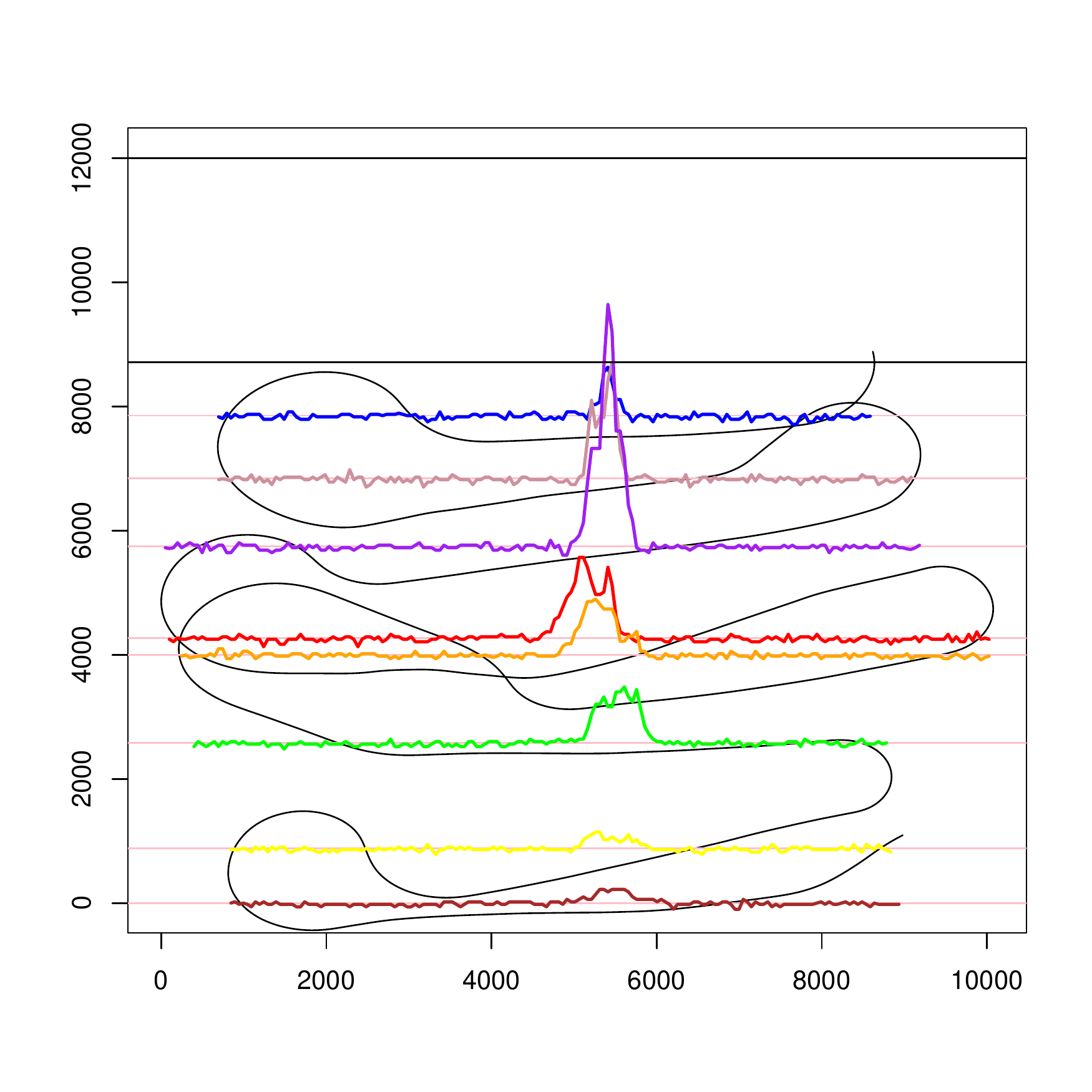}
\caption{The modified landfill data of the left-hand trajectory}
\label{Left_TransectConcentrations}
\end{figure}

Note that to avoid confounding of gas seepage when crossing the actual landfill, we only consider the data collected within the blue and red highlighted trajectory regions. 
An alternative view of the left trajectory of the data is provided in Figure \ref{Left_TransectConcentrations} plotting the methane concentrations when aligned to reference distances from the source.  {\cred The concentration data is collected discretely in time, resulting in a time series with varying length of around 200 data points in each leg as plotted in Figure \ref{Left_TransectConcentrations}. We re-register this data to form a multivariate time series with regularly spaced observations (for details we refer to ~\cite{silke_diss}, Chapter 18) before applying our methodology.}

Earlier work exploring this data set, described by \cite{Hirst2013}, sought to combine the observed gas concentration rates with idealised gas dispersion models to identify the locations of the unknown sources. Whilst effective, the method proposed requires strong assumptions on both the form of the (Gaussian) plume and about the dependence structure along the observed time series, in the form of independent and identically distributed Gaussian errors. 

In this article we develop an alternative approach that both allows for dependence in error structure along the flight path, and makes less restrictive assumptions on the plume form. Specifically, we seek to develop theory and methodology that  enable the analysis of such multivariate series, allowing for both dependence in time and a functional relationship between the location of change points across different components of the time series. We propose two different methods: The first only requires such a functional relationship generalizing the approach by \cite{horvath1999testing}, while the second one also uses approximate information about the reduction in concentration as the distance from the source increases. The latter approach has the potential to greatly increase power and hence estimation accuracy (see, e.g., \cite{aston2018}), while still being sufficiently robust with respect to a certain degree of misspecification of this concentration reduction.

The intuition that underpins our work is to view each of the aircraft transects as a time series in its own right, see Figure \ref{Left_TransectConcentrations} for an example. As such, our data is converted into a multivariate time series with each component corresponding to a transect of the flight path. 
Assuming that a given time series component (transect) includes a crossing of the plume, then one would expect to see an elevated concentration of gas in the time region that corresponds to the aircraft crossing the plume, with lower concentrations either side of the plume. Henceforth we shall refer to this region of elevated gas concentration as the \emph{change region}. In the statistical literature, situations where the mean in  an unknown interval differs from the rest of the data are called epidemic change problems  (see e.g. \citet{eeg_data}, \citet{aston2018}).
The feature that sets the gas emission data apart from other epidemic change situations is the fact that the locations are not at the same place in each component.  
Instead, due to the dispersion of the gas, it is natural to assume that the boundaries of the change regions are related to one another. The methodology which we propose seeks to encapsulate this relationship, allowing for a functional relationship that is parametrized by both known parameters (such as wind direction) and unknown parameters, e.g. the location of the source. 

The article is organised as follows. In Section \ref{sec:cpa:model}  
we give a general model description that is well suitable for the gas emission data after an appropriate preprocessing, but also allows for different examples. Section \ref{sec:cpa:test} derives and analyses two types of change point tests for the described model. While they may be of independent interest in other applications, for the purpose of the analysis of the gas emission data they are merely required as an intermediate step. In Section~\ref{section_est_cp} we derive two different estimators for the unknown source location (or more generally for the unknown parameters of the functional relationships describing the change region) and prove their consistency. Section~\ref{section_summary} summarizes the construction principles behind these tests and estimators and gives some insight into possible generalizations. Some simulations are given in Section~\ref{sec:sim_study}, while the left trajectory of the gas emission source is analyzed in detail in Section~\ref{sec:data}. Some concluding remarks can be found in Section~\ref{sec:conclusions}.
The proofs can be found in Appendix~\ref{sec:proofs} and the analysis of the right trajectory in Appendix~\ref{sec_right}.

 \section{Change point analysis}\label{sec:cpa}


 In this section, we begin by first describing a {\cred multivariate} modelling framework that takes the various attributes of the remote sensing change point problem into account.  From this we {\cred propose two different ways of aggregating information across transects that will be the basis for the proposed estimators for the location of the gas emission source. Because estimation and testing are strongly related we also provide the corresponding test procedures in Appendix~\ref{app_test}.} 
 
 In both cases, the developed theory goes beyond the motivating data example of gas emission sources.
Nevertheless, we will make the connection to the data at hand at every step, while discussing the  underlying construction principles and their consequences in more detail 
 in Section~\ref{section_summary}. 
 In so doing, we seek to 
better understand how to customize or even generalize the presented procedures to other situations and examples.

 \subsection{Model of the data}\label{sec:cpa:model}

As described in the introduction, following an appropriate transformation, the data is considered as a (dependent) multivariate time series with a different (in this example elevated) mean, within the change region of each component.
We define the change region $\{NF_{\vth_0}(i)<t\ls NG_{\vth_0}(i)\}$ in component $i$ by  a pair of change points (in rescaled time) $(F_{\vth_0}(i),G_{\vth_0}(i))$, with $F_{\vth}(i)<G_{\vth}(i)$ for all $i=1,\ldots,d$ and all $\vth\in\Theta$. Here, $\vth_0\in\Theta$ denotes the true underlying parameters, while the functional relationship between change points is  parametrized  by the functions $F_{\vth}(\cdot)$ and $G_{\vth}(\cdot)$. 
Clearly, these functions depend on both known parameters, such as the direction of the wind, and unknown parameters such as the location of the source and the opening angle of the cloud. For notational simplicity we will include the known parameters in the functional shape of $F,G$, so that $\vth\in\Theta$ are the unknown parameters only.

This leads to the following model for the data 
 \begin{equation}\label{eq_model}
	X_i(t) = \mu_i + {\Delta}_i \mathds{1}_{\{F_{\vth_0}(i) <t/N\ls G_{\vth_0}(i)\}} + e_i(t),
\end{equation}
with $i=1, ..., d$ denoting the components of the multivariate time series and $t=1, ..., N$ the time point (after transforming the flight path into a multivariate time series).
Furthermore, we assume that $\vth\mapsto F_{\vth}(i)$ as well as $\vth\mapsto G_{\vth}(i)$ are continuous for all $i=1,\ldots,d$.
The errors $\{\boldsymbol{e}(\cdot)\}$ with $\boldsymbol{e}(t)=(e_1(t),\ldots,e_d(t))^T$ are stationary and centered with existing second moments and have to fulfill a (multivariate) functional central limit theorem. In particular, they can be dependent.

This model extends the classical epidemic setting, where $\vth=(\lambda_1,\lambda_2)$, $0<\lambda_1<\lambda_2<1$ are the two unknown change points (in rescaled time) and $F_{\vth}(i)=\lambda_1$, $G_{\vth}(i)=\lambda_2$.

\begin{figure}[bt]
	\begin{center}
	{\includegraphics[width=5in]{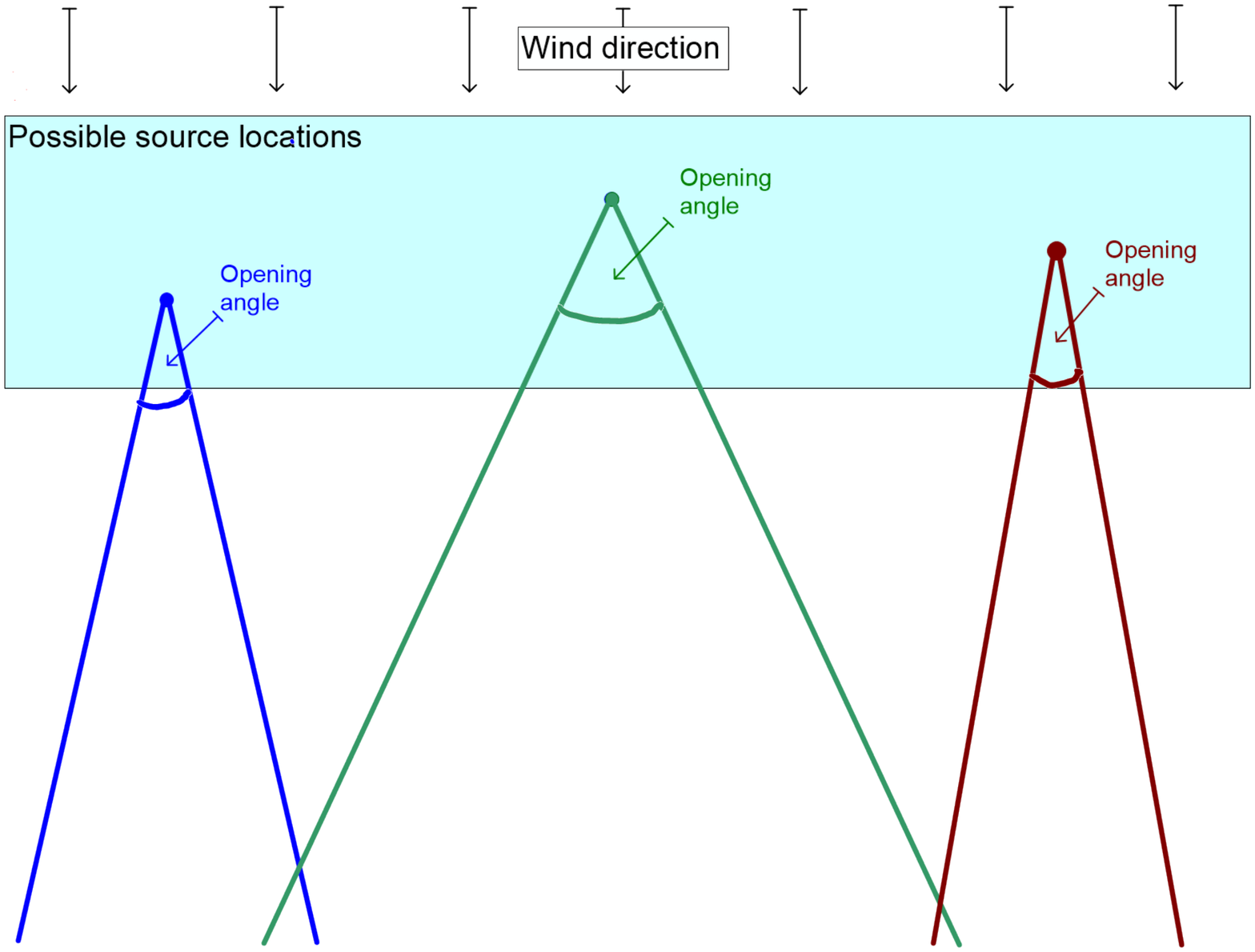}}
\end{center}
\caption{{\cred Linear plume model: Three exemplary source locations with a linear cloud having different opening angles.}}
\label{figure_lin_plume}
\end{figure}

{\cred The main example in this paper is a linear plume with known or unknown opening angle as shown in Figure~\ref{figure_lin_plume}.  The shaded field indicates possible source locations to be searched while we  indicate three possible clouds with different source locations and different opening angles. The wind direction is not included in this model because the information is already taken into consideration at the time of the collection of the data, where the flight paths are chosen to be perpendicular to the wind direction.

{\revtwo	The data consists of an 8-dimensional time series with only around 200 time points. Consequently, slightly different plume shapes  will lead to almost the same change points in each of the transects. Indeed, some  preliminary analyses have shown that both linear and Gaussian plumes lead to very similar results for the data example at hand.  As such, in order to aid clarity and model parsimony, we adopt  the simplest reasonable model in the simulation study and data analysis, namely a linear plume.
 The theoretic results obtained under model~\eqref{eq_model} are much more general and go far beyond the linear plume example by allowing for many different shapes of the cloud.

 Nevertheless, this model remains somewhat simplistic in other respects such as the assumption of a constant mean within the plume/change region, while the real data rather exhibits a gradual change.  The methodology in this paper could easily be adapted to those type of changes using the same tools by similar adjustment to that for the projection method. However, this  only leads to an improvement if the shape of the gradual change is known sufficiently well, which is typically not the case. 

 Additionally, 
 the alignment of the change points in the data example seems to deviate somewhat from any of the usual cloud shapes by being somewhat misaligned from one transect to the next, possibly caused by temporal changes in the wind direction, in particular for the right transect. We make use of this last observation by checking the robustness of our methodology  with respect to  misspecification; see Section~\ref{sec_right} in the appendix.

}
\subsection{Aggregation methodology}\label{sec:cpa:test}

{\cred 
In a multivariate context aggregating information about possible change points in different components of the time series usually leads to an improved signal-to-noise ratio. This is due to the fact that by the aggregation the signal is increased by a larger amount than the noise level  as long as the errors are not perfectly dependent. Consequently, a multivariate approach is usually preferable over several univariate approaches that are then combined later.

Therefore, we consider two different methods of aggregating information across transects that will be used for estimation purposes but can also be used for testing, as detailed in Appendix~\ref{app_test}.
}

The first {\cred approach} is related to the multivariate test statistic discussed in \citet{horvath1999testing} in the at-most-one-change situation {\cred which is obtained as a version of the likelihood ratio test statistic under  normality assumptions}. Their statistic is strongly related to the panel statistic as discussed in \citet{horvath2012change}, where the difference lies in the fact that the number of components can be similarly large or even larger than the number of time points (requiring different asymptotic considerations). 
 Our statistic is different because it (a) takes an epidemic change into account and, more importantly, that (b) we allow for general parametrizations of how the change evolves through components (by allowing for an arbitrary parametrization of the change points). 

 {\cred The multivariate approach we propose is based on the following statistic}
 \begin{align*}
	 & 
	 {\cred A^M(\vth) = }  \bS_{\vth}^T \Sigma ^{-1} \bS_{\vth},\\
	 &	 \text{where}\quad  \bS_{\vth}=(S_{\vth}(1),\ldots,S_{\vth}(d))^T,\qquad
	 S_{\vth}(i) = \sum\limits_{t=\lfloor N F_{\vth}(i)\rfloor+1}^{\lfloor NG_{\vth}(i)\rfloor} \left( X_i(t) - \frac{1}{N} \sum\limits_{l=1}^{N} X_i(l) \right),\\
	 &\phantom{\text{where}}\quad \Sigma=\sum_{h\in\ZZ}\Gamma(h), \quad \Gamma(h)=\E \boldsymbol{e}(0)\boldsymbol{e}(h)^T, h\gs 0,\quad \Gamma(h)=\Gamma(-h)^T, h<0.
 \end{align*}
 $\Sigma$ is the long-run covariance of the multivariate error sequence and can be replaced by a consistent estimator. In case of independent (across time) errors this reduces to the covariance matrix of $\boldsymbol{e}(0)$. If the dimension is even moderately large, the nonparametric estimation of the full long-run covariance matrix is statistically usually quite imprecise. {\cred See also the discussion in Remark~\ref{rem_mis_cov} in the appendix}. This is particularly problematic if the inverse of the covariance matrix is needed as is the case with the above statistic.
 {\cred  If the number of transects is large in comparison to the number of time points, then estimation errors can accumulate and identification may not be possible (\cite{BickelL2008}), where additional numerical errors may arise when inverting the matrix (see Chapter 14 in \cite{Higham2002}). The problem becomes even more difficult in the presence of time series errors (which requires the estimation of the spectrum at frequency 0) as well as under the presence of change points.}

 This is {\cred less problematic} if  $\Sigma$ has a diagonal structure, i.e.\ if the components are independent, and only the long-run variances have to be estimated. In our example, this assumption is reasonable (see Figure~\ref{figure_ACF_1}) otherwise bootstrap methods such as e.g.\ in \citet{Aston2012} can help. 
 {\cred Because the dependence between different transects seems to be very small (see Figure~\ref{figure_ACF_1} below), the latter approach is feasible even without using bootstrap methods. }
 {\cred Nevertheless, because of these difficulties  we also discuss the theoretic behavior of the testing and estimation procedures under misspecification i.e.\ allowing for inconsistent estimation of $\Sigma$ towards some positive definite matrix  $\Sigma_A$ that is not necessarily the true (long-run) covariance matrix of the errors.}

 {\cred
Additionally, the} estimation of the covariance matrix in a change point situation is complicated by the contamination by the change. For this reason, it is necessary to use the estimated errors within the estimation procedure.

{\cred Whilst the theory developed is completely general with respect to the choice of estimators of $\Sigma$, in the simulations and data example we use the following estimator:} Similarly to \cite{Aston2012}  the  errors are estimated componentwise  by
{\allowdisplaybreaks\begin{align*}
&\hat{e}_i(t) = X_i(t) - \widehat{\mu}_i - \widehat{\Delta}_i \mathds{1}{\left\{\widehat{f}_i < t \leq \widehat{g}_i \right\}},\\
&\text{where }\widehat{\mu}_i = \frac{1}{\widehat{f}_i + N - \widehat{g}_i} \left( \sum_{t=1}^{\widehat{f}_i} X_i(t) + \sum_{t=\widehat{g}_i+1}^{N} X_i(t) \right),\\
&\phantom{\text{where}}
\widehat{\Delta}_i = \frac{1}{\widehat{g}_i - \widehat{f}_i} \sum_{t=\widehat{f}_i+1}^{\widehat{g}_i}  X_i(t) - \hat{\mu}_i, \\
&\phantom{\text{where}}
\left( \widehat{f}_i,\widehat{g}_i \right) 
= \arg\max \left\{ \left| \sum_{t=f_i + 1}^{g_i} \left( X_i(t) - \overline{X}_{i,N} \right) \right|: 1 \leq f_i < g_i \leq N \right\}.
\end{align*}}

{\cred In the dependent case, } we estimate the long-run variances $\widehat{\sigma}_i^2$, $i=1,\ldots,d$, by the flat-top estimator with automatic bandwidth selection as proposed by \cite{politis2003adaptive} 
{\cred based on the estimated residuals}.
The long-run covariance matrix {\cred is then estimated} by the corresponding diagonal matrix $\widehat{\Sigma}=\mbox{diag}(\widehat{\sigma}_1^2,\ldots,\widehat{\sigma}_d^2)$.

{\cred The above way of aggregating is optimal if the change vector $\boldsymbol{\Delta}:=(\Delta_1,\ldots,\Delta_d)^T$ is allowed to be completely arbitrary. Often additional structural assumptions about $\boldsymbol{\Delta}$ are being made such as e.g.\ sparsity in the sense of many zeros. In such situations, many different approaches exist based on the idea of using a suitable projection into a lower dimensional space. For example, ~\cite{wang2018high} use a sparse singular value decomposition, \cite{cho2015multiple} use thresholding, \cite{jirak2015} uses information for each component separately, 
\cite{mei2010efficient} and \cite{wang2018thresholded}
use a set of possibilities for which components are non-zero. A theoretical discussion of the potential of using (appropriate) projection methods in a multivariate setting can be found in~\cite{aston2018}.

 In our data example we may reasonably assume some knowledge about the (relative) decay of the concentration from one transect to the other.
This information has not been taken into account by the above multivariate statistic: More precisely, we may assume to have } information about the change direction $\boldsymbol{\Delta}/ \|\boldsymbol{\Delta}\|$, {\cred where $\|\cdot\|$} is the Euclidean norm. While the exact physical decrease depends on several parameters and is difficult to know precisely, at least a rough direction is known. Specifically, the concentration might first increase (keeping in mind that the plume is actually a 3D-object so that the plane might only run into it at some distance behind the source), but then it will drop.  This information can be used to increase the signal-to-noise ratio by using a projection onto $\widetilde{\boldsymbol{\Delta}}=(\widetilde{\Delta}_1,\ldots,\widetilde{\Delta}_d)^T$, which ideally is a multiple of $\mathbf{\Delta}$. In order to obtain the best signal-to-noise ratio, the data first needs to be standardized by $\Sigma_A^{-1/2}$, which also alters the change direction (hence the  projection direction)  by a factor $\Sigma_A^{-1/2}$. An ideal choice is given by $\Sigma_A=\Sigma$, where $\Sigma$ is the true covariance matrix, which is usually difficult to obtain, hence we allow for misspecification {\cred where $\Sigma_A\neq \Sigma$ in the below theory.}
 
 If $\widetilde{\boldsymbol{\Delta}}$ is close to a multiple of the true concentration direction, the signal-to-noise ratio will greatly improve resulting in {\cred more precise estimators as well as higher testing power} (see \citet{aston2018} for more details). Projecting onto $\widecheck{\boldsymbol{\Delta}}=\Sigma_A^{-1/2}\widetilde{\boldsymbol{\Delta}}/{\|\Sigma_A^{-1/2}\widetilde{\boldsymbol{\Delta}}\|}$ yields the projected time series $\{Y(\cdot)\}$ with
 \begin{align*}
	 &\|\Sigma_A^{-1/2}\widetilde{\boldsymbol{\Delta}}\|\,	 Y(t)=\boldsymbol{X}(t)^T\Sigma_A^{-1}\widetilde{\boldsymbol{\Delta}}=  \left(\boldsymbol{\Delta}_{\{F_{\vth_0}(\cdot)<t/N\ls G_{\vth_0}(\cdot)\}}\right)^T\Sigma_A^{-1}\widetilde{\boldsymbol{\Delta}}+\boldsymbol{e}(t)^T\Sigma_A^{-1}\widetilde{\boldsymbol{\Delta}},
\\
&\text{where}\quad \boldsymbol{X}(t)=(X_1(t),\ldots,X_d(t))^T, \quad \boldsymbol{\mu}=(\mu_1,\ldots,\mu_d)^T \text{ and}\\
&\boldsymbol{\Delta}_{\{F_{\vth_0}(\cdot)<t/N\ls G_{\vth_0}(\cdot)\}}=\left({\Delta}_{\{F_{\vth_0}(\cdot)<t/N\ls G_{\vth_0}(\cdot)\}}(1),\ldots,{\Delta}_{\{F_{\vth_0}(\cdot)<t/N\ls G_{\vth_0}(\cdot)\}}(d)  \right)^T,
\\
&\text{with }{\Delta}_{\{F_{\vth}(\cdot)<t/N\ls G_{\vth}(\cdot)\}}(i)=\begin{cases}
	\Delta_i, &F_{\vth}(i)<t/N\ls G_{\vth}(i),\\
	0, &\text{otherwise}.
\end{cases}
 \end{align*}
We define
\begin{align*}
	&	Y(t)=\mathbf{D}^{{\Delta},\widetilde{\Delta}}_{\vth}(t)+e_P(t),\\
&\text{where }\mathbf{D}^{\Delta,\widetilde{\Delta}}_{\vth}(t):= \frac{\left(\boldsymbol{\Delta}_{\{F_{\vth_0}(\cdot)<t/N\ls G_{\vth_0}(\cdot)\}}\right)^T\Sigma_A^{-1}\widetilde{\boldsymbol{\Delta}}}{\|\Sigma_A^{-1/2}\widetilde{\boldsymbol{\Delta}}\|},\qquad {e}_P(t):=\frac{\boldsymbol{e}(t)^T\Sigma_A^{-1}\widetilde{\boldsymbol{\Delta}}}{\|\Sigma_A^{-1/2}\widetilde{\boldsymbol{\Delta}}\|}.
\end{align*}

By the above assumptions the function $\vartheta \mapsto \mathbf{D}^{{\Delta},\widetilde{\Delta}}_{\vartheta}(\cdot)$ is continuous and $t \mapsto \mathbf{D}^{{\Delta},\widetilde{\Delta}}_{\vartheta}(t)$ is left-continuous.

In the classical epidemic setting with $\vth=(\lambda_1,\lambda_2)$ and $F_{\vth}(i)=\lambda_1$, $G_{\vth}(i)=\lambda_2$, the projection time series also has an epidemic change. However, in the general model~\eqref{eq_model}  it exhibits a gradual (epidemic) change (see Figure \ref{abrupt_vs_gradual_change}). More precisely, if $\widetilde{\boldsymbol{\Delta}}$ is correct and a diagonal covariance matrix $\Sigma_A$ is used, then the signal part has the following shape (multiplied by a constant indicating the strength of the change) in rescaled time $s=t/N$:
 \begin{align*}
	 \mathbf{D}^{\widetilde{\Delta},\widetilde{\Delta}}_{\vth}(s)=\|\Sigma_A^{-1/2}\widetilde{\boldsymbol{\Delta}}\|\,\sum_{i=1}^d{\widecheck{\Delta}}_i^2{\mathds{1}}_{\{F_{\vth}(i)<s\ls G_{\vth}(i)\}}.
 \end{align*}
\begin{figure}[b]
\begin{subfloat}[Simulated multivariate data under $H_1$, i.e.\ with abrupt epidemic mean changes in every component.]
{\includegraphics[width=3in]{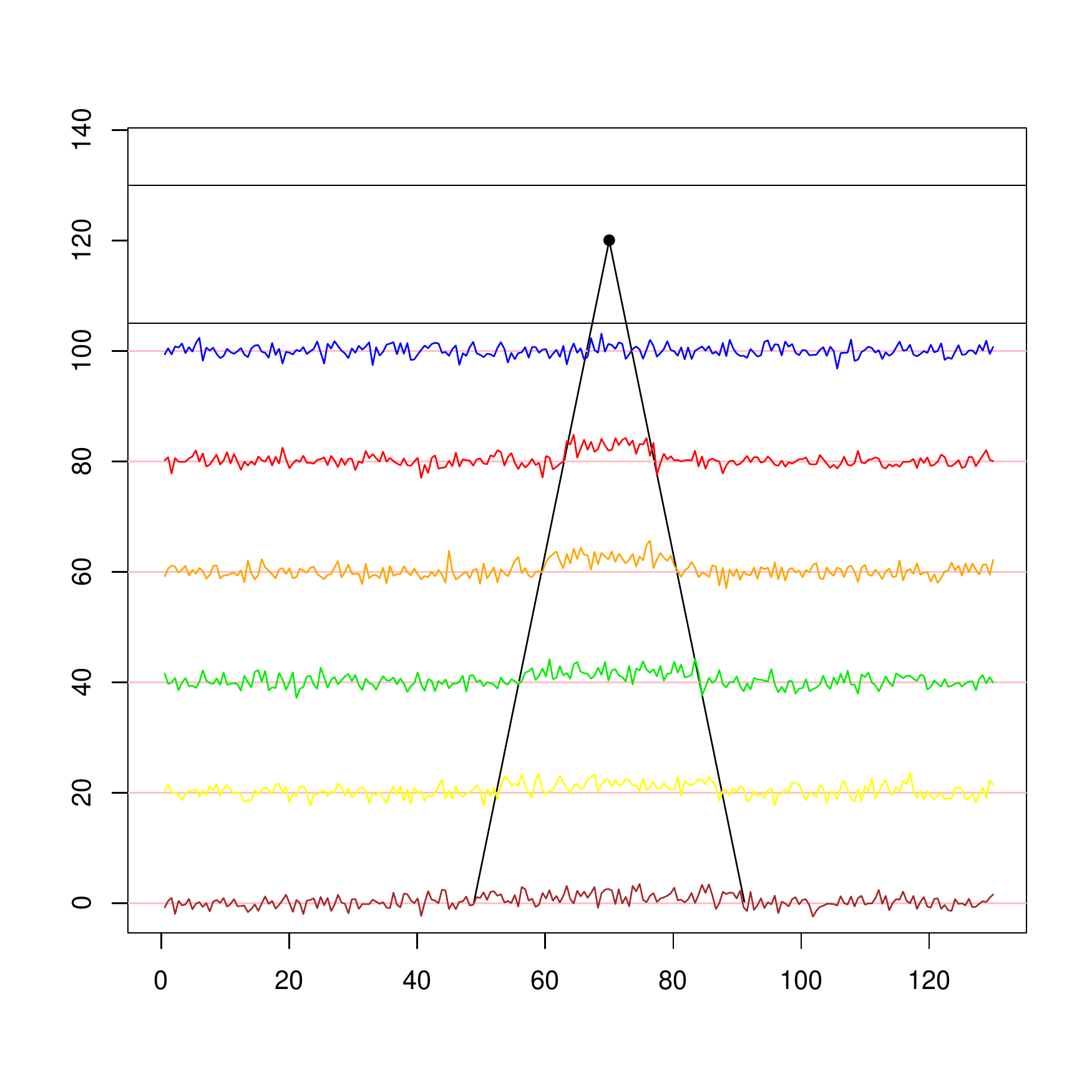}}
\end{subfloat}
\begin{subfloat}[Resalting univariate sequence with a gradual epidemic mean change.]
{\includegraphics[width=3in]{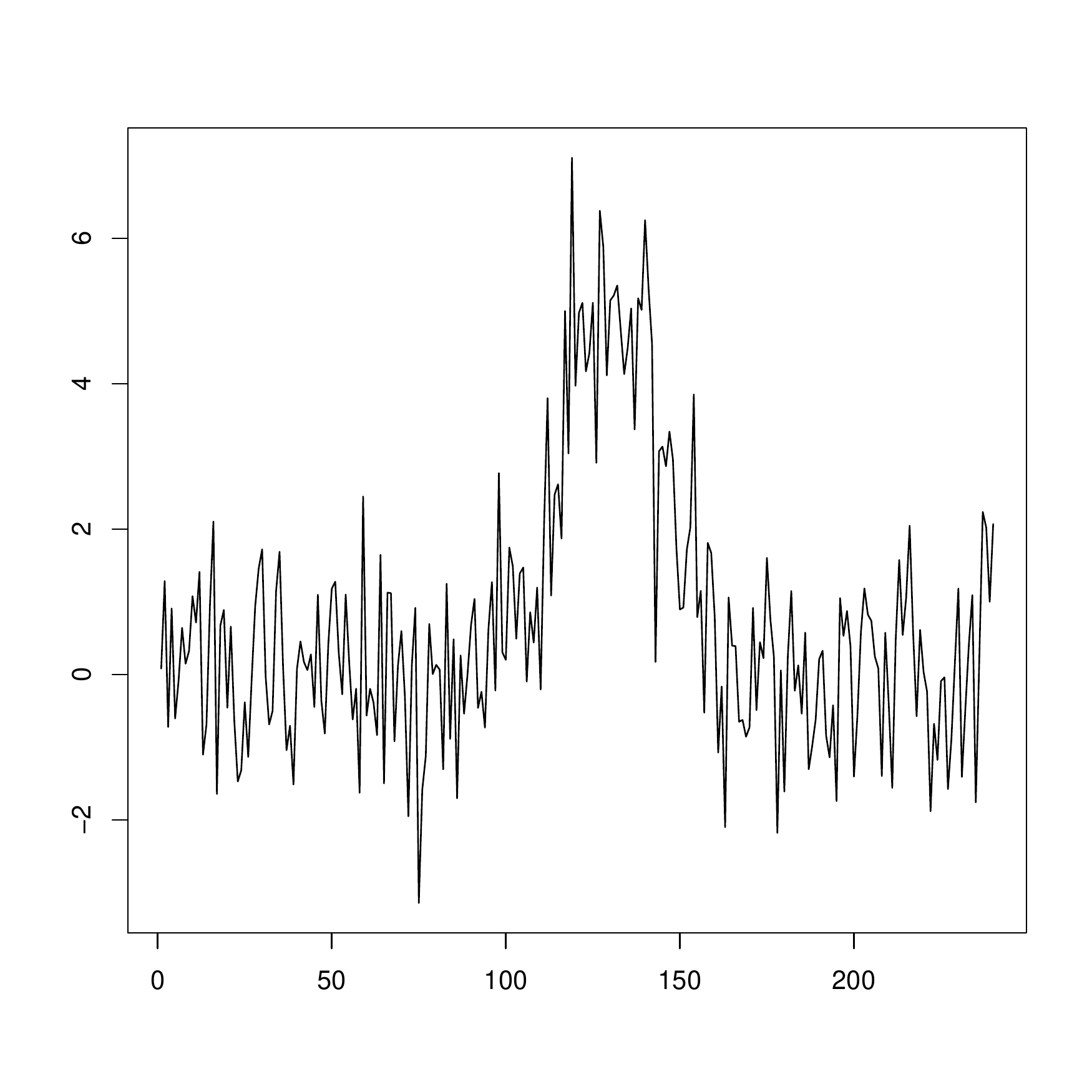}}
\end{subfloat}
\caption{Simulated multivariate data $\bs{X}(t)$ under $H_1$ and resulting univariate sequence $Y(t)$.}
\label{abrupt_vs_gradual_change}
\end{figure}

{\cred The projection-based aggregation thus results in} 
\begin{align*}
	{\cred A^P(\vth)} &= \left| \sum\limits_{t=1}^{N} \left( \mathbf{D}_{\vth}(t/N) -\frac{1}{N} \sum\limits_{l=1}^{N} \mathbf{D}_{\vth}(l/N) \right) Y(t) \right|
	 =  \left| \sum\limits_{t=1}^{N}  \mathbf{D}_{\vth}(t/N) ( Y(t)-\bar{Y}_N) \right|,
\end{align*}
where $\mathbf{D}_{\vth}=\mathbf{D}_{\vth}^{\widetilde{\Delta},\widetilde{\Delta}}$. {\cred We note that this statistic is related to the one by \citet{extreme_gradual} and \citet{limitdistr_gradual} that was obtained as the likelihood ratio statistic for a (non-epidemic) gradual change with a given polynomial slope.}

{\cred
	Based on these two version of aggregating information across different transects we derive estimators for the source location in the next section.  In the context of change point detection there is a strong connection between estimators and tests in the sense that often test statistics are obtained as the maximum over all possible parameters, $\vth$, while estimators are obtained as the point $\vth$ that  maximizes the corresponding test statistic. Also if a test statistic has a large power for a given alternative then the corresponding estimator will typically be more precise.  Therefore, in Appendix~\ref{app_test} we detail properties of the corresponding test statistics, part of which are also needed to prove consistency of the corresponding estimators.

}

\subsection{Estimation of change points/gas emission source}\label{section_est_cp}

 In classical  change point procedures, such as the ones discussed in the previous section,  natural estimators for the location of the change point can be obtained by looking at the point where the maximum is obtained.
 Similarly, in the setting which we consider, the parameter maximizing the statistic is an estimator for the true parameter value:
 \begin{align*}
	 &	 \widehat{\vth}_M={\cred\arg\max_{\vth\in\Theta}A^M(\vth)}=\arg\max_{\vth\in\Theta}\bS_{\vth}^T \Sigma ^{-1} \bS_{\vth},\\
&\widehat{\vth}_P=
{\cred\arg\max_{\vth\in\Theta}  \frac{A^P(\vth)}{\left(\sum_{t=1}^N\left(\mathbf{D}_{\vth}(t/N) -\frac 1 N\sum_{l=1}^N\mathbf{D}_{\vth}(l/N)\right)^2  \right)^{1/2}}}\\
&=\qquad
\arg\max_{\vth\in\Theta}  \frac{\left| \sum\limits_{t=1}^{N}  \mathbf{D}_{\vth}(t/N) ( Y(t)-\bar{Y}_N) \right|}{\left(\sum_{t=1}^N\left(\mathbf{D}_{\vth}(t/N) -\frac 1 N\sum_{l=1}^N\mathbf{D}_{\vth}(l/N)\right)^2  \right)^{1/2}},
 \end{align*}
 where $\arg\max$ is the set of all maximizing values. {\cred In practise some representative is used. The normalization of the projection estimator is necessary in order to obtain consistent results for the source estimation in this gradual (after projection) situation.}
 In particular we obtain an estimate for a plume, where the true source can be expected to be close to the origin of that plume. 
 {\cred However, identifiability in a small sample situation can be weak: E.g.\ in the gas emission example only relatively few data points per transect are observed so that many different clouds will cut each transect at almost the same locations. While each of the corresponding clouds segments the data reasonably, the actual source locations may vary by a much larger margin. In the data example, this effect can be seen by looking at the heat maps  in  Figures \ref{left_traj}(c) and (d) as well as \ref{right_traj_2} where the value of the statisic is very similar along vertical 'lines' in the source area. In this case, a source location higher up in combination with a slightly smaller opening angle results in a very similar segmentation of the data.

 }

 {\cred As pointed out in Section~\ref{sec:cpa:test} correct estimation of $\Sigma$ in particular in a time series/change point context may be difficult. Therefore, we explicitely allow for misestimation in the below theorem by 
letting $\widehat{\Sigma}$ converge to some matrix $\Sigma_A$ that can be different from $\Sigma$.}
 \begin{theorem}\label{theorem_est}
	 Let the assumptions on the errors of Theorem~\ref{theorem_null} hold. Furthermore, choose $\Theta$ such that $\vth_0$ is identifiable, i.e.\ there does not exist $\vth_1\neq\vth_0\in\Theta$ such that $F_{\vth_1}(i)=F_{\vth_0}(i)$ as well as $G_{\vth_1}(i)=G_{\vth_0}(i)$ for all $i=1,\ldots,d$ with $\Delta_i\neq 0$. Then, under a fixed alternative as in \eqref{eq_model}  with $\Delta_i\neq 0$ for at least one $i=1,\ldots,d$, it holds:
\begin{enumerate}[(a)]
	\item If $\widehat{\Sigma}\pto \Sigma_A$ for some diagonal positive definite matrix $\Sigma_A$ (not necessarily equal to $\Sigma$), the estimators based on the multivariate statistic are consistent, i.e.\ $\widehat{\vth}_M\pto \vth_0$.
	\item Let the true parameter be identifiable from the projected signal in the sense that there does not exist $\vth_1\neq \vth_0$ such that $\mathbf{D}_{\vth_1}=a \mathbf{D}_{\vth_0}+b$ for some constants $a,b$. Then, if the projection direction $\widetilde{\Delta}$ (but not necessarily $\Sigma_A$)  and the cloud shape are correct, the estimators based on the projection statistic are consistent $\widehat{\vth}_P\pto \vth_0$.
\end{enumerate}
 \end{theorem}
 For a linear plume and diagonal $\Sigma_A$, as is assumed in the data example and  simulation study, the identifiability condition in the above theorem holds as soon as there is a change in at least two components. However, as we shall see,  in the true data example with  small $N$  and varying $y$-coordinate of the source location, the difference is very small. Consequently, it follows that the surface of $\bS_{\vth}^T \Sigma ^{-1} \bS_{\vth}$ is very flat along the $y$-axis.  This is clearly seen by the heatmaps (for the possible sources) of the statistics for the gas emission data example in  Figures \ref{left_traj}(c) and (d) as well as \ref{right_traj_2}.

 {\cred The following remark gives some additional insight into the effect of misspecification or misestimation of the covariance structure on the estimation of the source location.}
 \begin{rem}\label{rem_est}
	 \begin{enumerate}[(a)]
		 \item
	 The assertion for the multivariate procedure also holds for non-diagonal covariance matrices $\Sigma_A$ as long as the true source location is the unique maximizer of the signal $\|\Sigma_A^{-1/2}\mathbf{H}_{\vth}\|$ with $\mathbf{H}_{\vth}=(H_{\vth}(1),\ldots,H_{\vth}(d))^T$ and $H_{\vth}(i)=\Delta_i \,h_{\vth,\vth_0}(i)$ as in Lemma~\ref{lemma_alt}.
 \item If the cloud or the projection direction is misspecified, then the assertion for the projection statistic only holds in the sense that the best approximating source will be estimated (if identifiable unique), where the best approximating parameter $\vth_1$ is obtained as the maximizer of
	 \begin{align*}
		 \int_0^1\widetilde{\mathbf{D}}_{\vth_1}(z)\widetilde{\mathbf{D}}_{\vth_0}^{\Delta,\widetilde{\Delta}}(z)\,dz,\qquad \text{where }\,\widetilde{\mathbf{D}}=\frac{\mathbf{D}-\int_0^1\mathbf{D}(z)\,dz}{\left( \int_0^1(\mathbf{D}(z)-\int_0^1\mathbf{D}(s)\,ds)^2\,dz \right)^{1/2}},
	 \end{align*}
	 where for the misspecified cloud $\mathbf{D}_{\vth_0}^{\Delta,\widetilde{\Delta}}$ is the projected signal obtained from the correct cloud shape and change $\Delta$ when using the projection direction $\widecheck{\Delta}$, while $\mathbf{D}_{\vth_1}$ is the projected signal based on the supposed change direction and cloud shape that is also used in the statistic. 
	 \end{enumerate}
 \end{rem}

 \subsection{Another look at the construction of the estimators}
	\label{section_summary}
	The model described in the previous section along with the testing procedures and estimators  go far beyond this particular data example. For this reason, we will discuss the construction of the proposed estimators as well as their strengths and weaknesses in this section.  
 
The methodology developed in this work is of potential use in those change point situations where a reasonable parametrization of a functional  relationship between change points of different components is available, and which may depend on unknown parameters (like e.g.\ the precise shape of the cloud or its opening angle in the gas emission example). Below, we therefore shed some more light on how to customize or generalize the presented procedures to other situations and examples. 

In this work we have considered the case of epidemic changes, where in each component there are precisely two changes and the mean of the first and last section are equal. One can easily extend this to situations of at most one change, where an example is   a recession evolving through different branches of economy with the recession not hitting all of them at the same time.  In such cases,  the function $F_{\vth}=0$ needs to be set to be equal to zero in all procedures.

Both proposed procedures rely on the functional relationship $F_{\vth}$ as well as $G_{\vth}$, which must be specified in advance {\cred motivated by the given data set at hand. However, by allowing for additional unknown parameters in the procedure such as e.g.\ a range of opening angles in the gas emission example, various possible shapes can be incorporated into the analysis.}
Nevertheless, there are several problems attached to this: On the one hand a precise estimation is only possible if the true functional relationship is included in this scenario, so that one may want to use a large number of parameters $\vth$. On the other, this can make both estimation and testing much more difficult as the true signal can more easily be hidden beneath some random false signal. In testing, this means the quantiles of the null distribution may increase significantly, requiring a stronger signal for detectability. In estimation, precision may be lost due to random fluctuations. {\cred This effect can clearly be seen by comparing the upper panel in Figure~\ref{figure_ck_1} with the lower panel. All four pictures are based on the same signal strength but in the upper panel the true opening angle has been used while in the lower panel  a range of opening angles is considered.} This effect is related to the usual trade-off between parametric and non-parametric methods.

Looking more closely at the assumptions required for consistency of the estimators,  an identifiability assumption is needed in the sense that every parameter $\vth$ leads to a unique signal in terms of change points. For the projection statistic this requirement is stronger as the  projected signal contains less information than the original multivariate signal. On the other hand, there is usually also less noise, which is an advantage. In practice, this identifiability issue shows by having several different sets of parameters that yield almost the same value of the statistic. This is true, in particular, if the number of parameters is large, i.e.\ fewer assumptions on the functional relationship of the change points are made. For the heat maps for the statistic at several different source locations as e.g.\ given in the last two panels of Figure \ref{left_traj} this can be seen by the large areas where the value of the statistic is particularly high. This is due to the fact that the data can similarly well be approximated by several sets of parameters (all leading to different but similar cloud shapes).

Both proposed procedures require an estimation of the inverse of the covariance matrix which is typically challenging in practice. The corresponding change point estimation problem is usually quite robust with respect to this, as also suggested by our theoretical results under misspecification. However, the testing procedure may suffer greatly. This is particularly bad for the multivariate procedure, where the Brownian bridges in the limit distribution are no longer independent, and critical values obtained from the independence assumptions are no longer valid. The consequence would be potentially dishonest (both conservative or liberal) testing procedures. The projection test on the other hand is much more robust in this respect, as the size is unaffected by dependence between components but it may suffer some power loss.

The main difference between the two procedures discussed in this work is that the projection method makes use of more information, requiring some knowledge about the functional relationship of the change points but also their relative strengths in each component (the absolute strength $\|\Delta\|$ does not matter). In many situations, such as in our data example, where some information about the diffusion of the gas can be used, such knowledge is available. We note that this information is not used by the multivariate statistic and, as a consequence, the signal-to-noise ratio of the projected time series is better so that both the power and the estimation precision increases. On the other hand, problems can also arise if the relative strength of the change in each component is misspecified. However, in our simulations we found the estimator of the clouds to be quite robust with respect to some mild to moderate misspecification of the relative strength of the change in each component. In the present paper, we have only worked with a  precisely known decay, as an alternative one could consider to let it depend on unknowns as well.

\section{Simulations and data analysis}
 \subsection{Some simulations}\label{sec:sim_study}
 In this section, we illustrate the small sample properties of the above procedures by a small simulation study.

\begin{figure}[b]
\begin{subfloat}[Multivariate $\widehat{\theta}_M$]
{\includegraphics[width=0.48\textwidth]{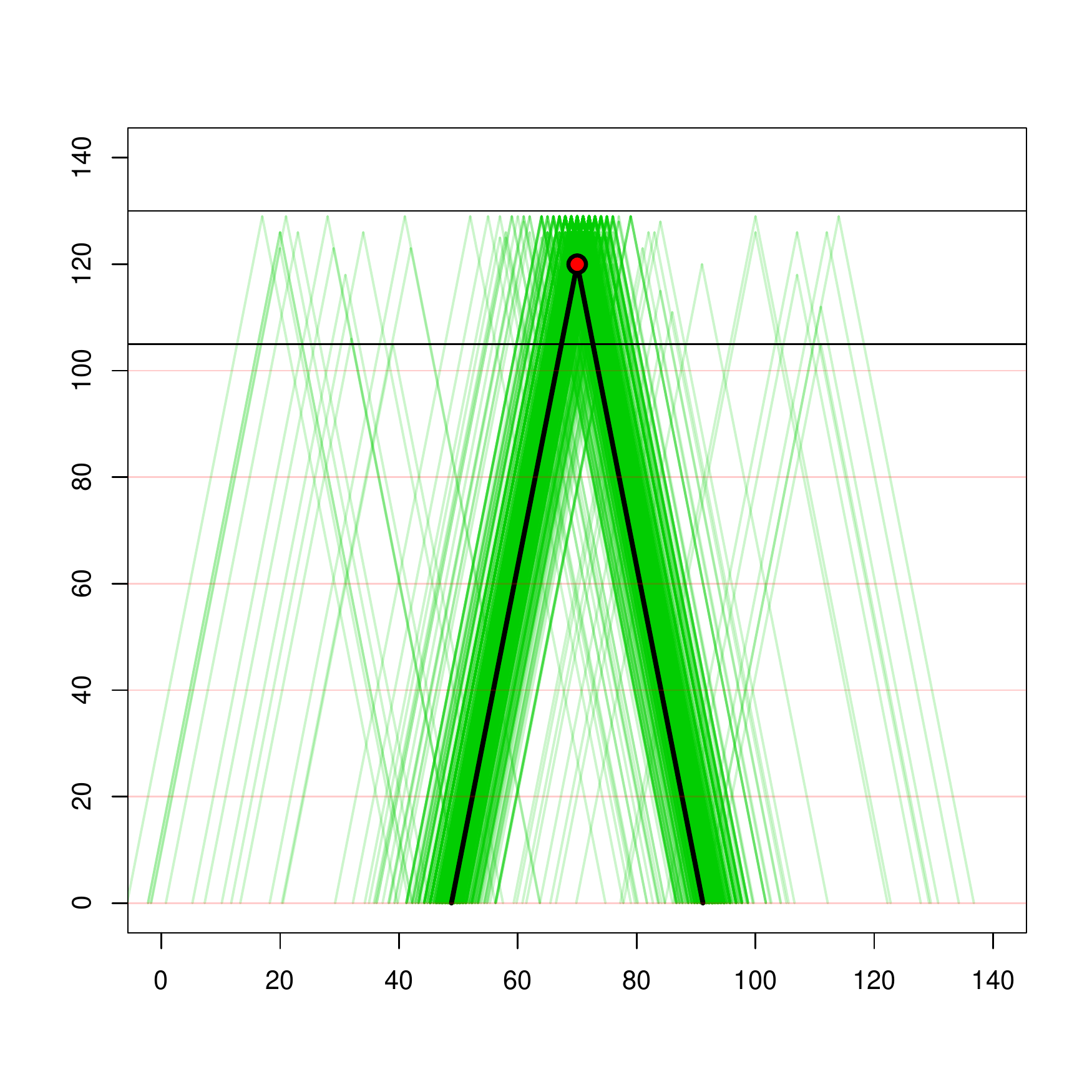}}
\end{subfloat}
\begin{subfloat}[Projection: $\widehat{\theta}_P$]
{\includegraphics[width=0.48\textwidth]{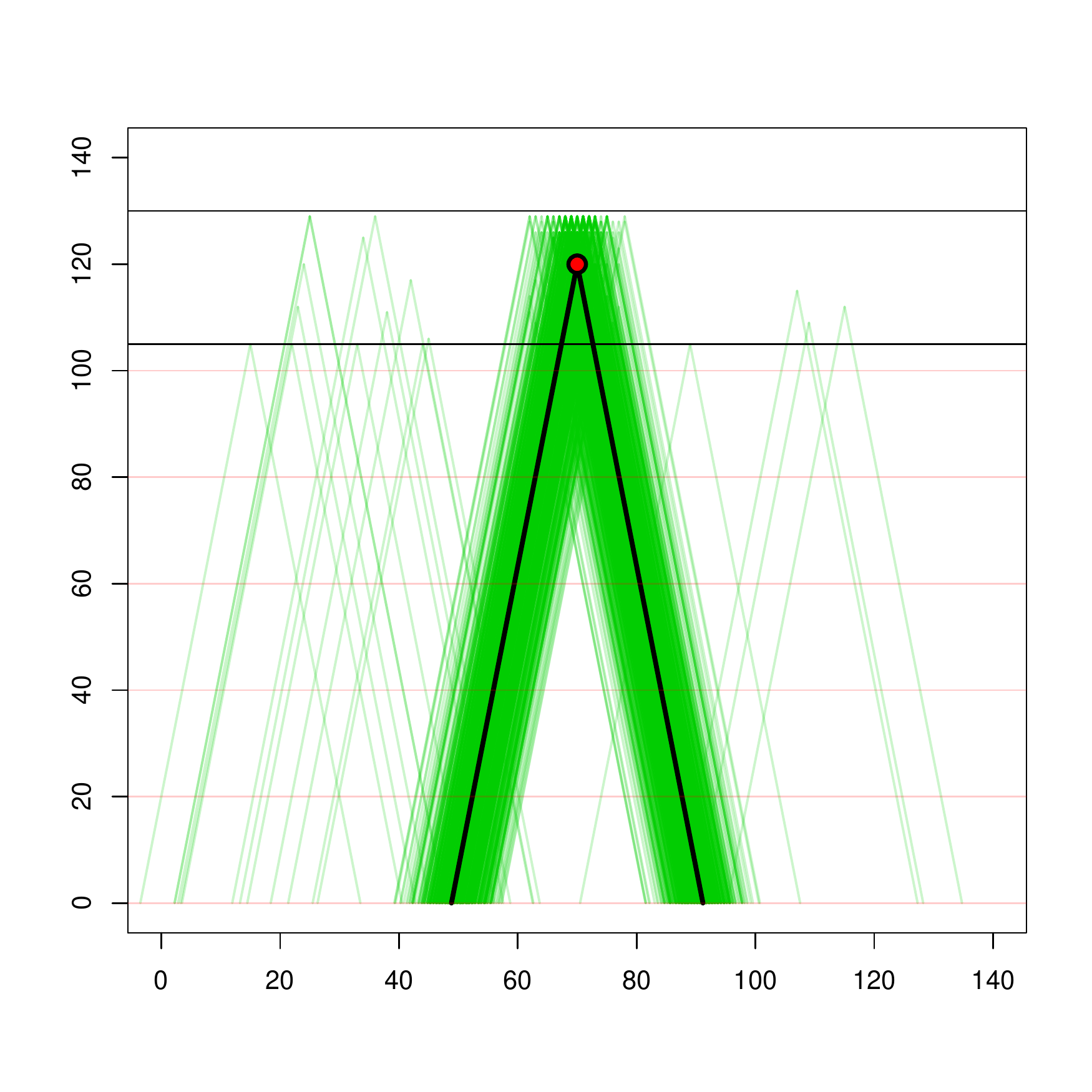}}
\end{subfloat}
\begin{subfloat}[Multivariate $\widehat{\theta}_M$]
	{\includegraphics[width=0.48\textwidth]{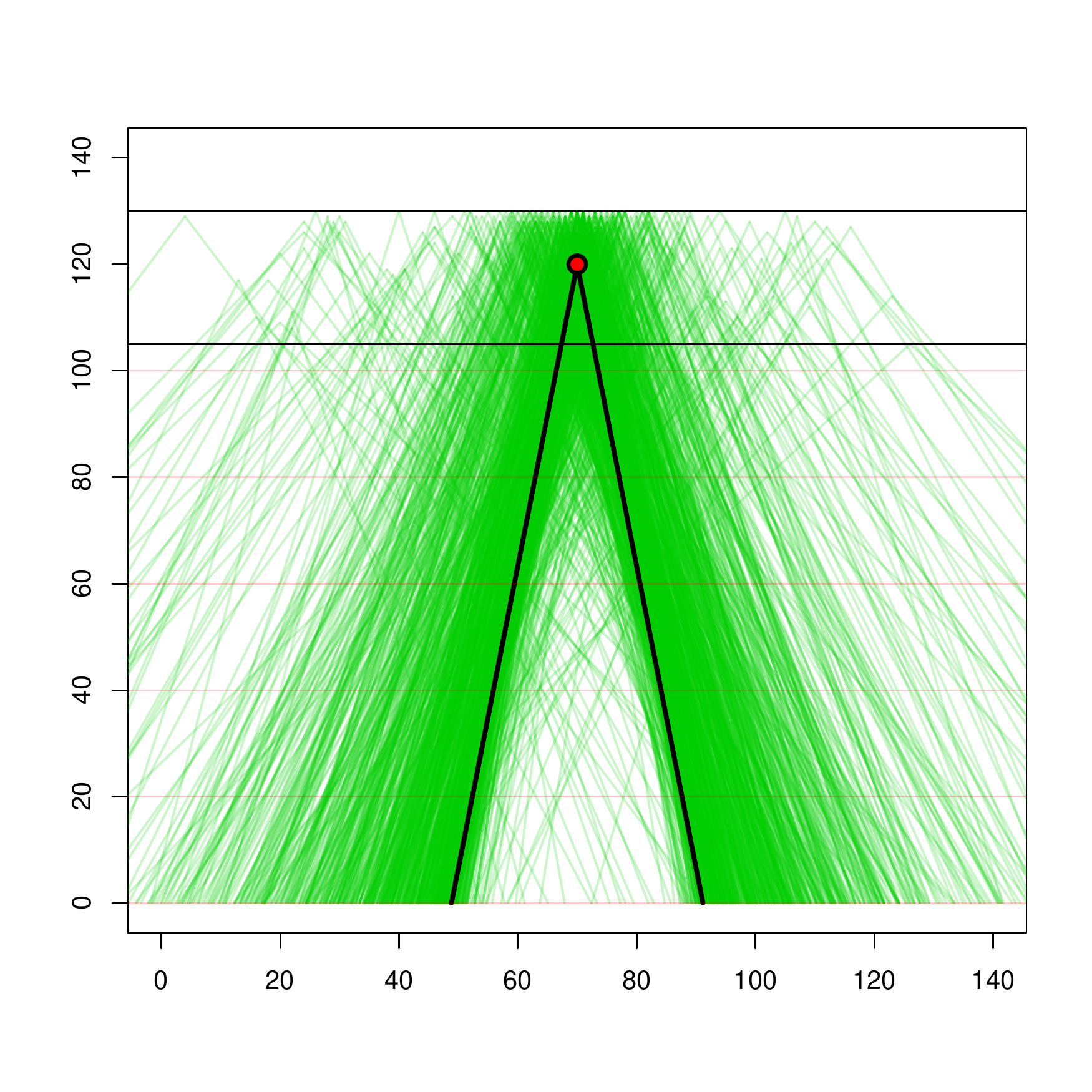}}
\end{subfloat}
\begin{subfloat}[Projection: $\widehat{\theta}_P$]
	{\includegraphics[width=0.48\textwidth]{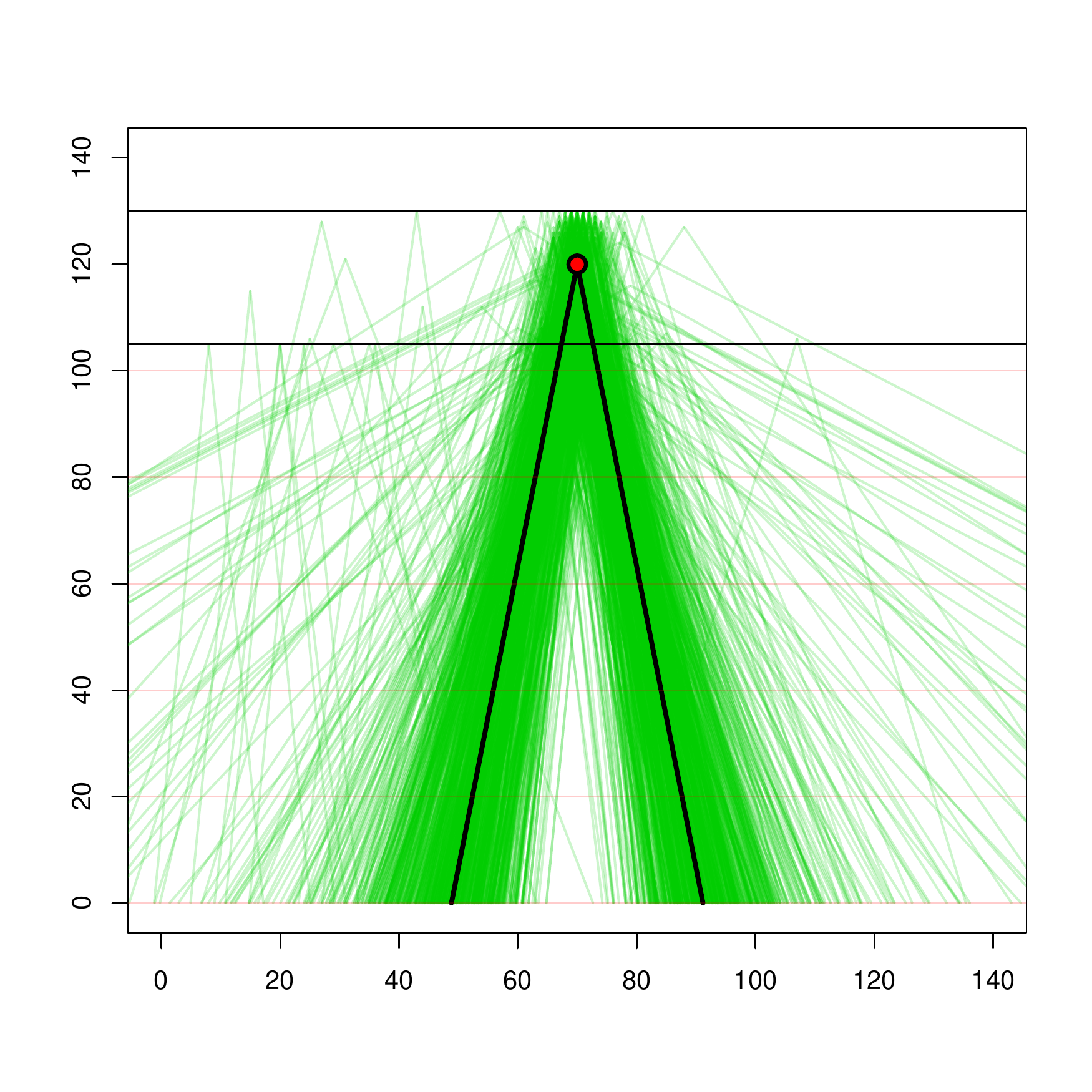}}
\end{subfloat}
\caption{Estimated clouds {\cred for i.i.d.\ data. Upper panel: Fixed opening angle  of 20 degrees, lower panel: Allowing for opening angles between $10$ and $120$ degrees.}}
\label{figure_ck_1}
\end{figure}

\begin{figure}[bp]
\begin{subfloat}[ $\tau=0.1$]
{\includegraphics[width=0.48\textwidth]{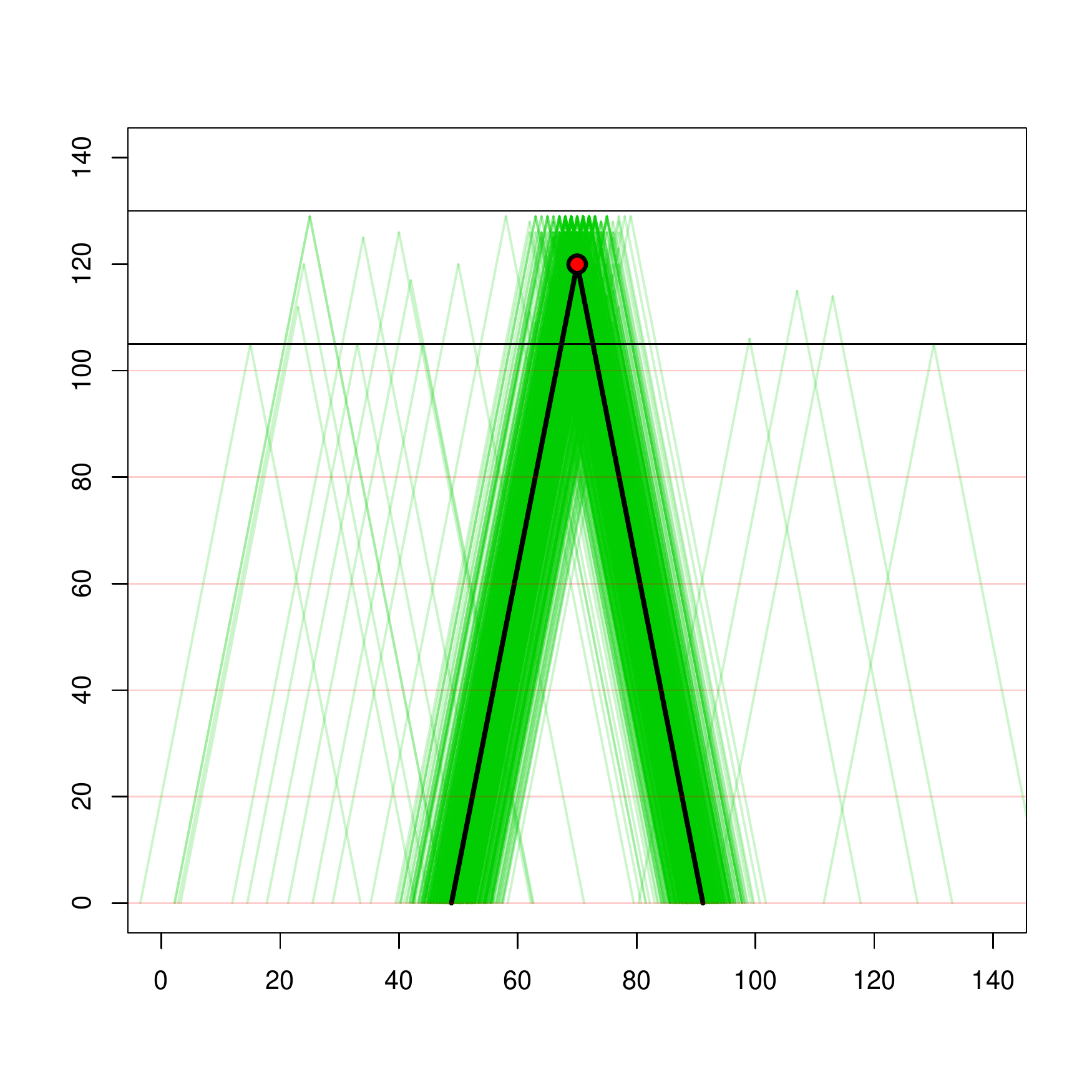}}
\end{subfloat}
\begin{subfloat}[$\tau=0.3$]
{\includegraphics[width=0.48\textwidth]{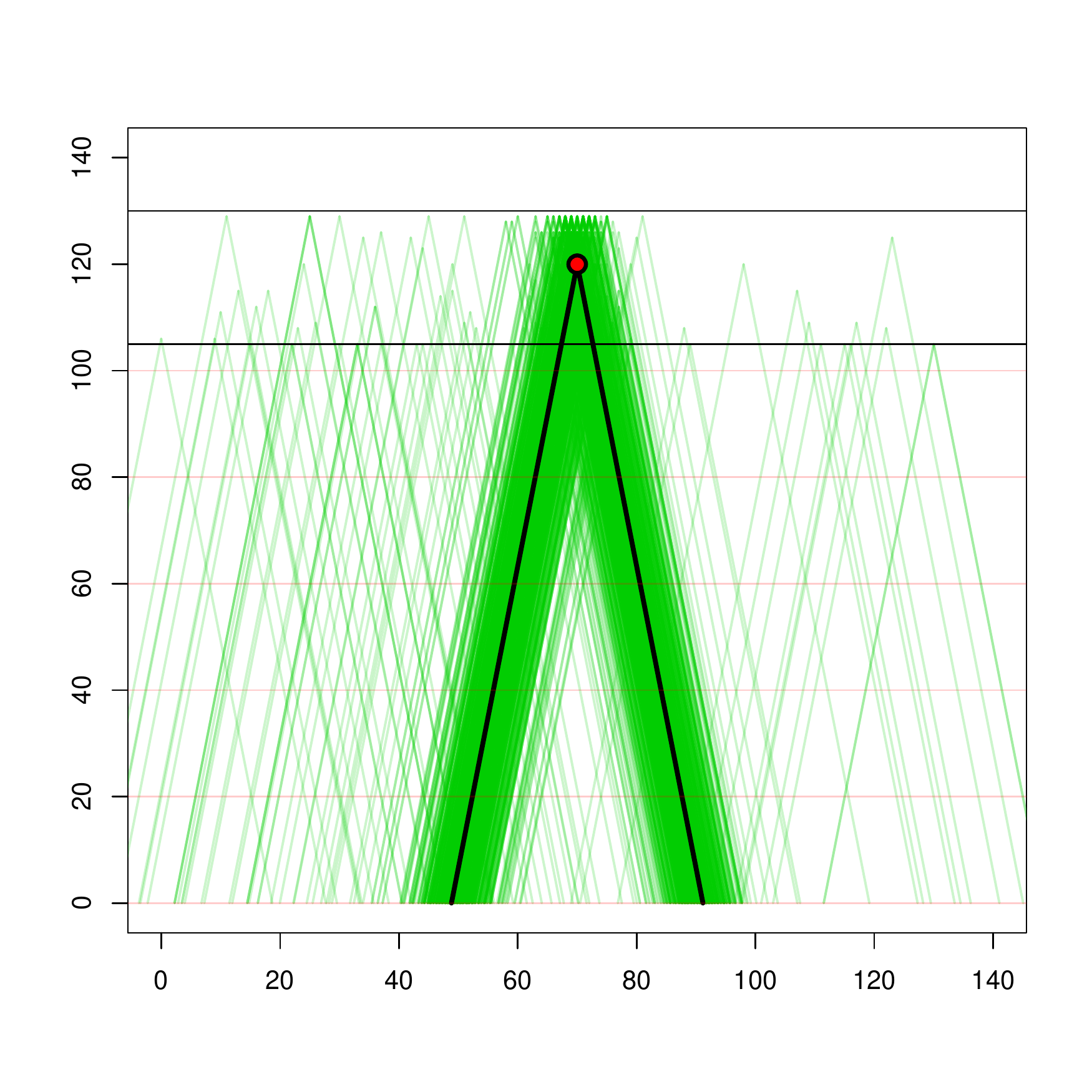}}
\end{subfloat}
\caption{Estimated clouds {\cred for the projection statistic under misspecification of the change direction by normal errors with standard deviation $\tau$}. The same signal strength as in Figure~\ref{figure_ck_1} has been used.}
\label{figure_ck_2}
\end{figure}

 {\cred We first focus on simulations supporting the theoretical observations in Section~\ref{section_summary}, which can best be seen for independent and identically distributed errors and by using the true variances. In a second step, time series errors are simulated with a similar autocorrelation structure as the estimated residuals from the data example. Additionally, the long-run covariances are estimated as described in Section~\ref{sec:cpa:test}, so that the simulated data is treated in exactly the same way as the gas emission data.  All simulations are based on Gaussian data.} Under the alternative, an epidemic mean change whose boundaries develop according to a linear cloud (with an opening angle of $20^\circ$), $d=6, N=240$, {\cred is simulated}. 

 {\cred The magnitude of the change is generated as follows:}
A plane flying in a certain height over the cloud will only enter it fully at a certain distance to the source keeping in mind that the cloud is a 3-D-object. Consequently, we simulate the magnitude of the change points $\Delta_j$, $j=1,\ldots,d,$ such that it first increases quickly before decreasing again at a slower rate. This effect can also be clearly seen in the data (see Figure~\ref{Left_TransectConcentrations}).  

 More results dealing with different signal strengths and weight functions as well as size and power of the corresponding test procedures can be found in Section 17 of \citet{silke_diss}.


{\cred Our main aim is to judge the quality of the estimated cloud including the variability of the estimator. The source itself may only be very weakly identifiable because there are only relatively few data points at each transect. Thus, clouds belonging to several different potential source points may lead to similar change point locations in each of the transects. This effect can also be seen by the vertical lines in the heatmaps of the data example (Figures~ \ref{left_traj}(c) and (d) as well as \ref{right_traj_2}). At each of those potential source locations clouds with varying opening angles exist which have a similarly good fit to the data. For this reason, we visualize the quality of the estimated clouds instead by plotting the estimated clouds from all $1000$ simulations in one plot together with the true cloud. }

Figure~\ref{figure_ck_1} (a) and (b) give the results under the assumption of a known fixed opening angle. 
The estimators from the projection statistic are {\cred somewhat} more precise than from the multivariate procedure, {\cred i.e.\ there are fewer estimated clouds at the wrong locations}. As discussed in Section~\ref{section_summary} the projection statistic -- unlike the multivariate statistic -- uses the additional information about the direction of the change $(\Delta_1,\ldots,\Delta_d)^T$ (up to multiplicative constants indicating the strength of the signal). In the gas emission example, this corresponds to having knowledge about the relative decline of the gas as the airplane gets further and further away from the source.

In Section~\ref{section_summary} it was also discussed how precision of the corresponding estimators ({\cred as well as power of the corresponding test statistics}) can be diminished by allowing for more flexibility in the parameters defining the cloud. We will check this effect empirically by not working with a fixed known opening angle but rather treat the angle as another unknown quantity. Effectively, this means that we are no longer only maximizing over the source location but also over the opening angle resulting in a substantial increase in computational effort.  {\cred The corresponding simulation results are given in Figure~\ref{figure_ck_1} (c) and (d). In this case, the estimators for the cloud become much less precise if applied to the same time series, such that a stronger signal is needed to obtain the same precision. Indeed, while the strength of the signal remains the same by allowing for more flexibility in modelling by means of an unknown opening angle, the noise level of the statistic is greatly increased, where clearly the multivariate statistic is affected more strongly.
 Thus,} both methods have, as one might expect, a much diminished quality of estimation.  

 {\cred This situation also shows that the gain in precision from the use of the projection statistic can be substantial } due to the use of the additional information of  the {\cred change direction i.e.\ the} decay of the gas concentration with distance from source. 
 
 {\cred In order to  check for robustness of the projection procedure with respect to misspecification of that change direction, we contaminate the true signal strength in each transect} by i.i.d.\ normal disturbances, i.e.\ the true change in component $i$ is given by $\Delta_i+\eps_i$, $\eps_i\sim N(0,\tau^2)$ i.i.d., $\|\mathbf{\Delta}\|=1$, while the projection statistic is still constructed with $\Delta_i$. {\cred We use $\tau=0.1$ as well as $\tau =0.3$, which is already a substantial contamination because the magnitude of the change in each transects lies between $0.22$ and $0.62$ so that the contamination is of similar magnitude than the signal.} The results  are given in Figure~\ref{figure_ck_2} showing that the procedure is indeed quite robust with respect to at least slight to medium deviations from the truth.

\begin{figure}[b]
\begin{subfloat}[Multivariate $\widehat{\theta}_M$]
{\includegraphics[width=0.48\textwidth]{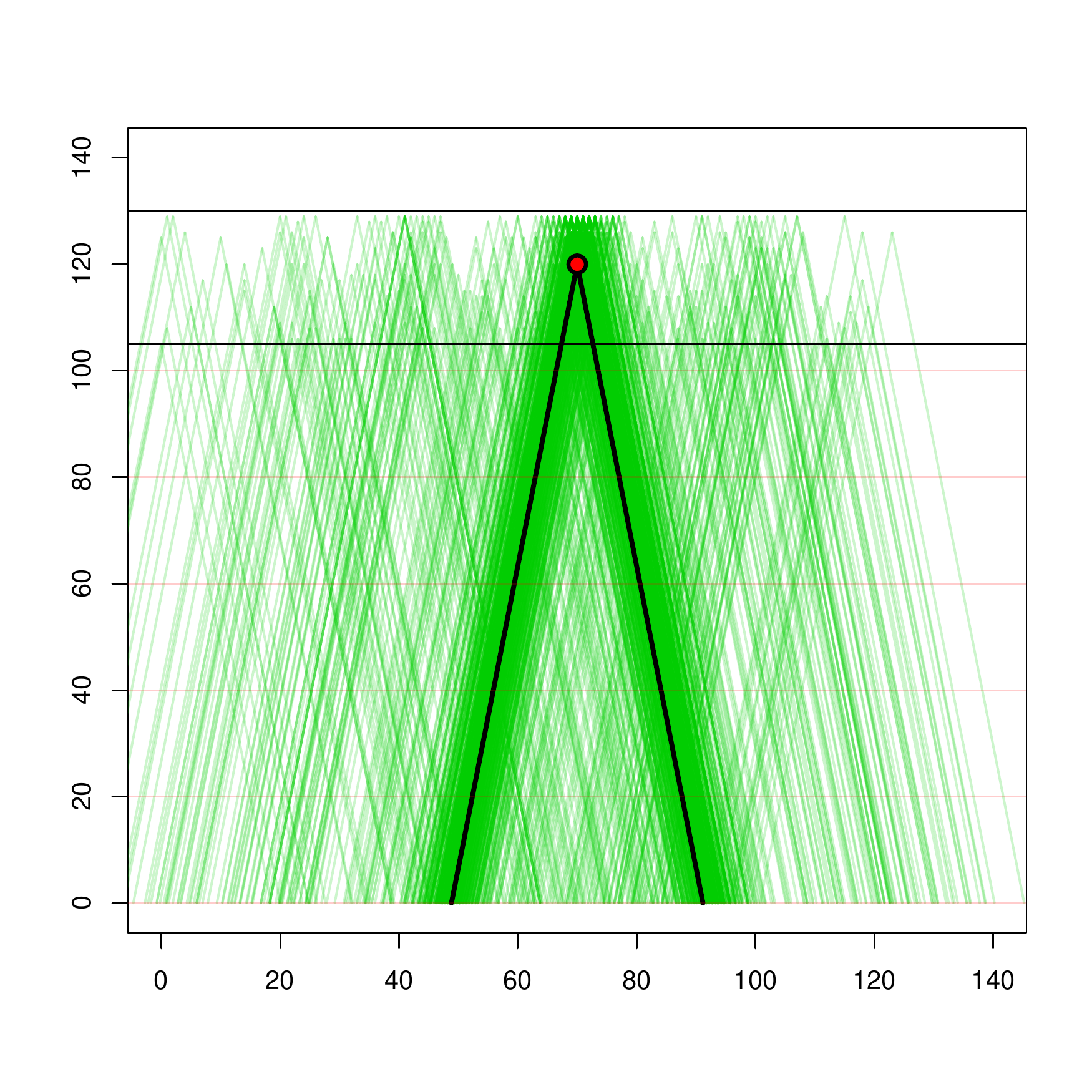}}
\end{subfloat}
\begin{subfloat}[Projection $\widehat{\theta}_P$]
{\includegraphics[width=0.48\textwidth]{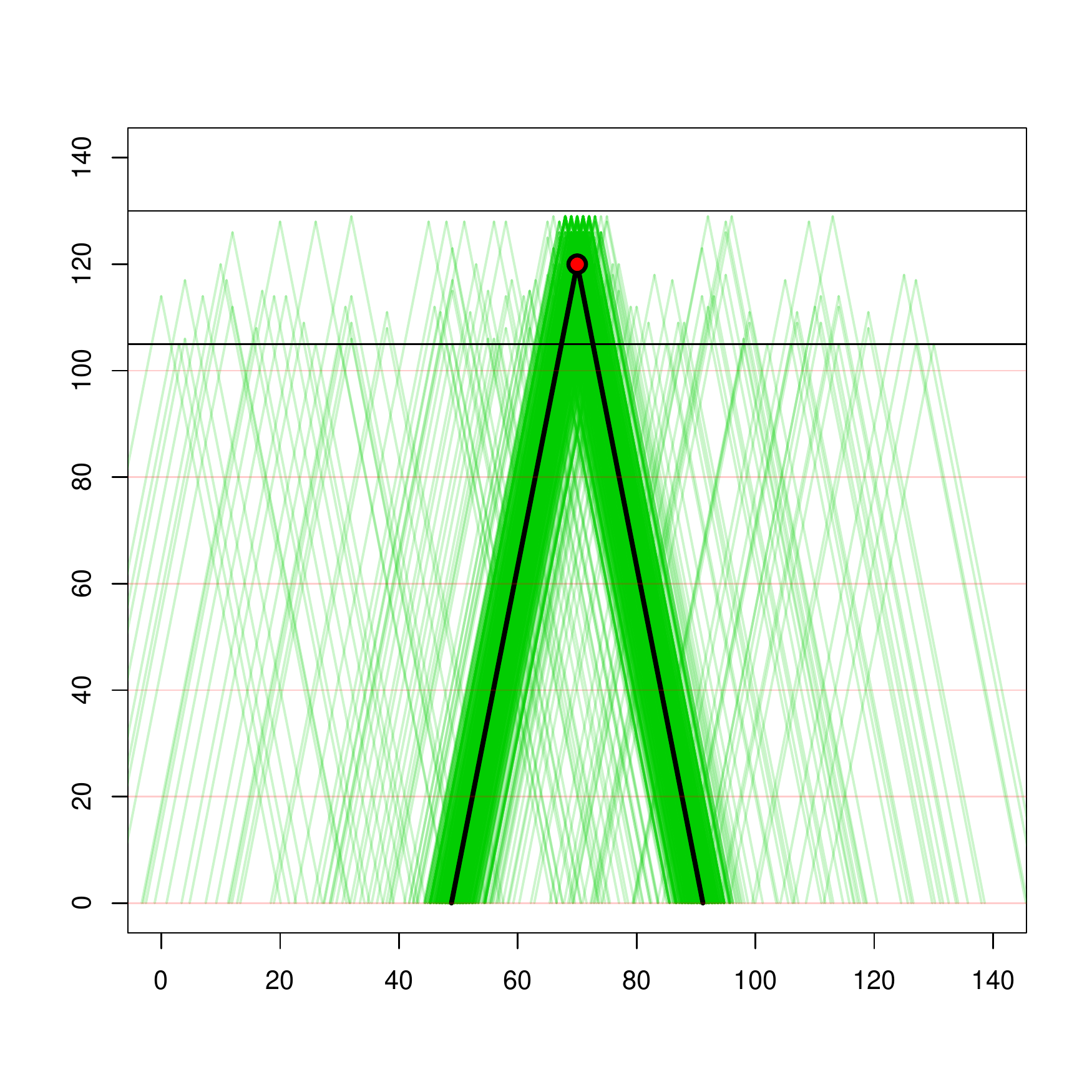}}
\end{subfloat}
\begin{subfloat}[Multivariate $\widehat{\theta}_M$]
{\includegraphics[width=0.48\textwidth]{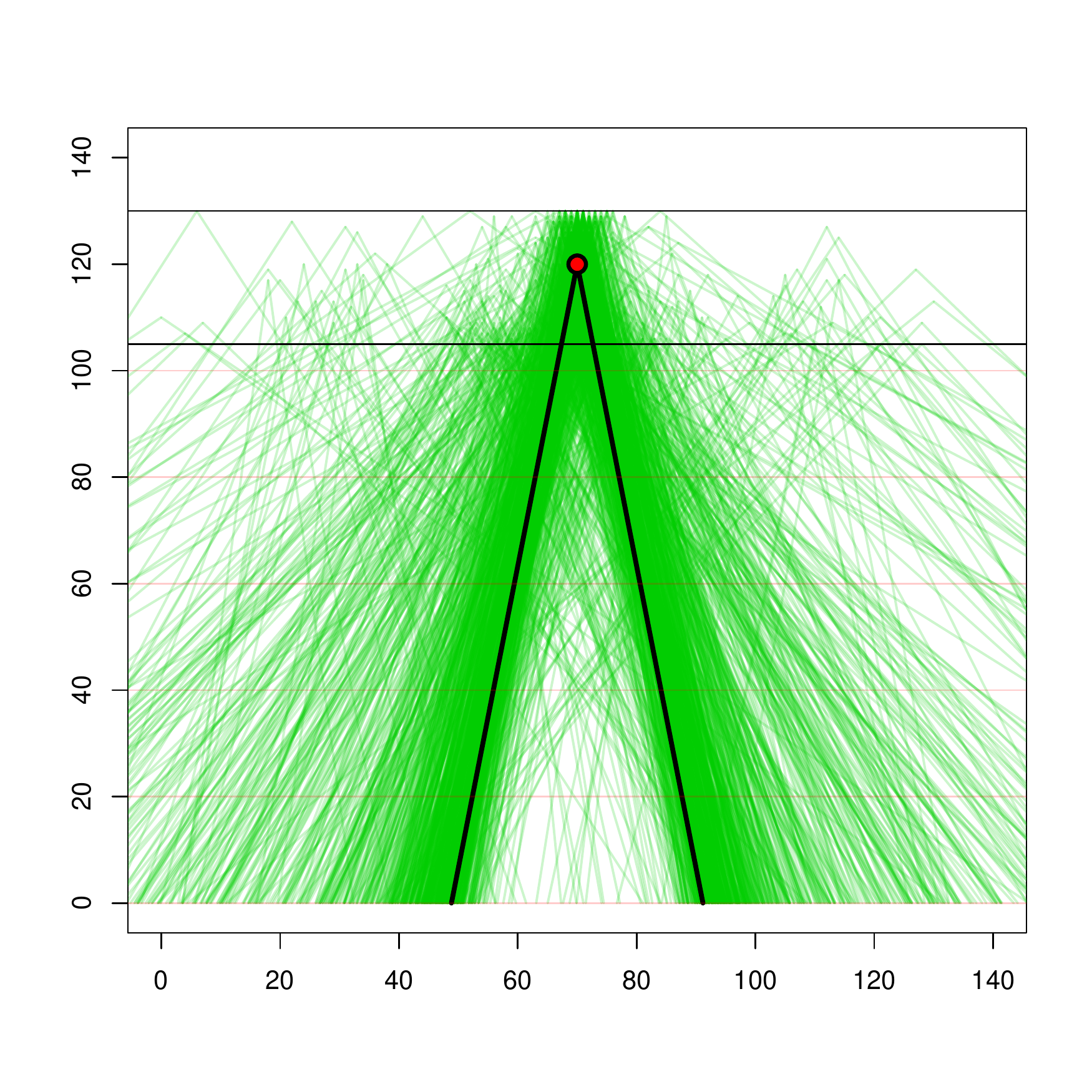}}
\end{subfloat}
\begin{subfloat}[Projection $\widehat{\theta}_P$]
{\includegraphics[width=0.48\textwidth]{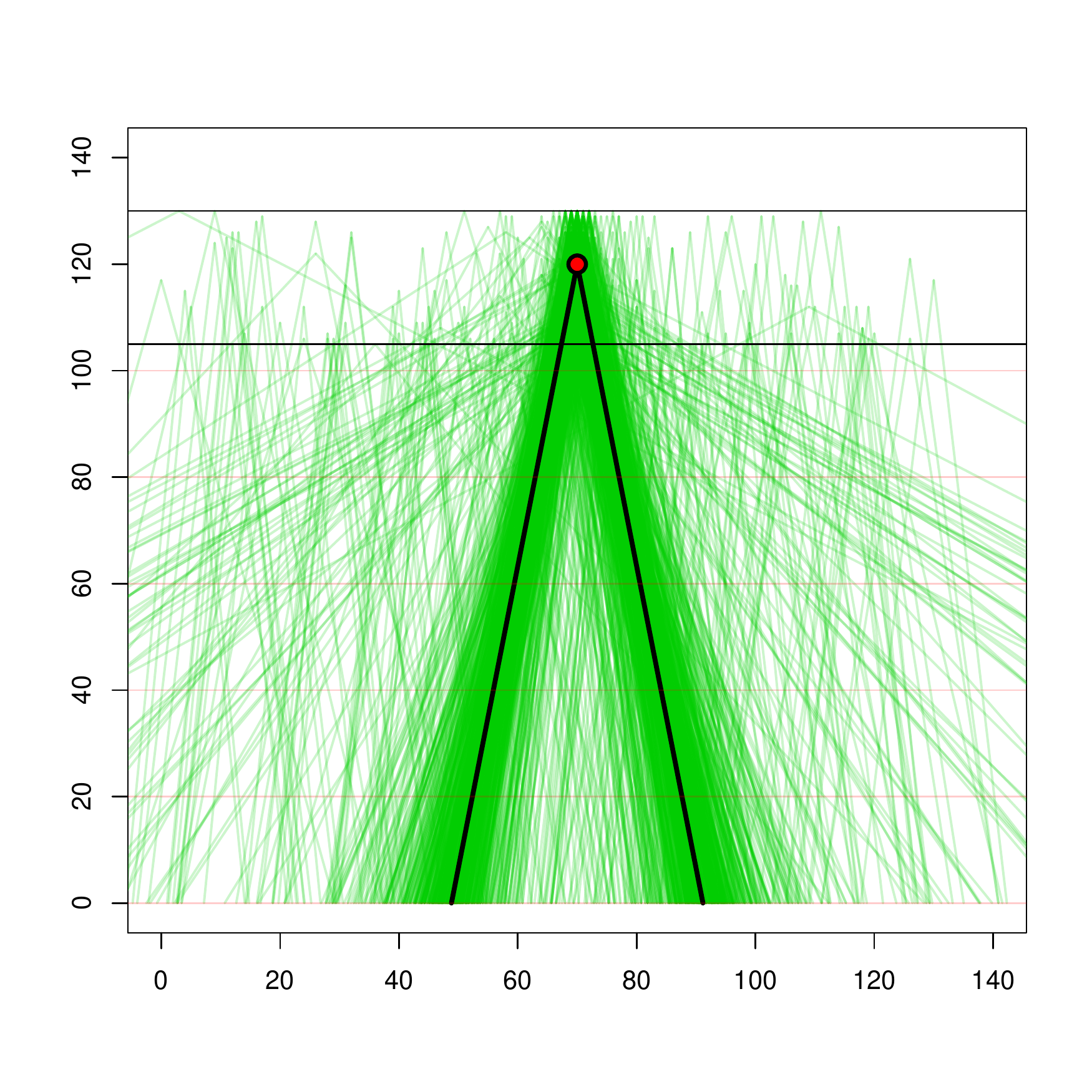}}
\end{subfloat}

\caption{Estimated clouds for dependent data: {\cred Upper panel: Fixed opening angle of $20$ degrees, lower panel: allowing for opening angles between $10$ and $120$ degrees}}
\label{figure_dep}
\end{figure}

{\cred
	In order to assess the effect of dependence on the  procedure, we use the following dependent model: Each transect is generated independently as the following MA(9) model with standard Gaussian white noise: $X_t=e_t+0.3 e_{t-1}+0.2e_{t-2}+0.1 e_{t-3}-0.1 e_{t-5}-\ldots -0.5 e_{t-9}$, which was chosen because its autocorrelation and partial autocorrelation structures look similar to what we have seen in the actual data (see Figure~\ref{figure_ACF_1} for the corresponding plots for the first transect).

	Because the noise level in this model is higher than for the above independent case and because we estimate the long-run variances, we use a stronger signal of $\|\mathbf{\Delta}\|=3$.  	The results can be found in Figure~\ref{figure_dep}. Clearly, the precision is again better for the projection than the multivariate procedure. As before, precision is better if the true opening angle is known as opposed to having to estimate the opening angle as well.
}

 \subsection{Data analysis}\label{sec:data}

 \begin{figure}[t]
\begin{subfloat}[ACF: First leg]
{\includegraphics[width=0.24\textwidth]{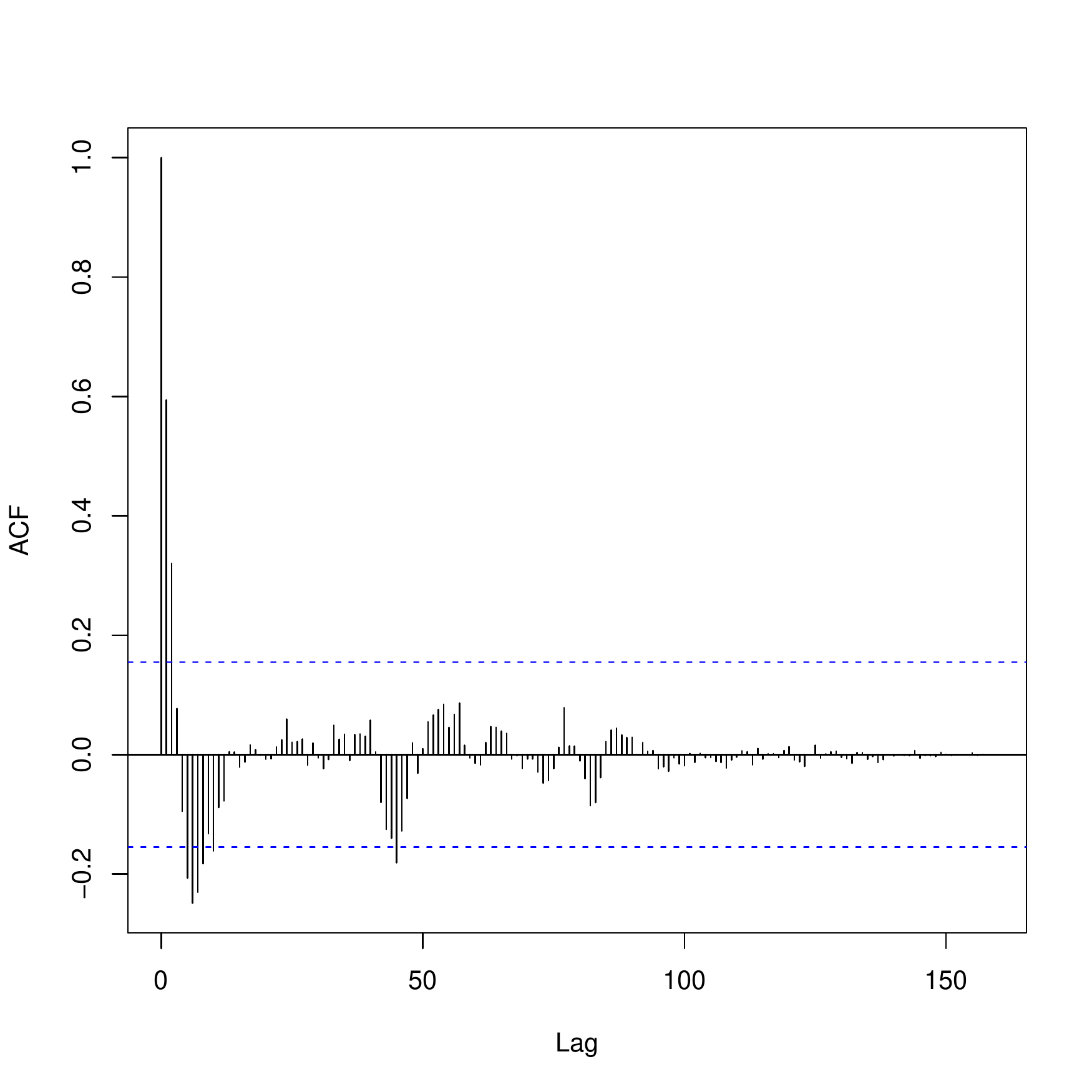}}
\end{subfloat}
\begin{subfloat}[PACF: First leg]
{\includegraphics[width=0.24\textwidth]{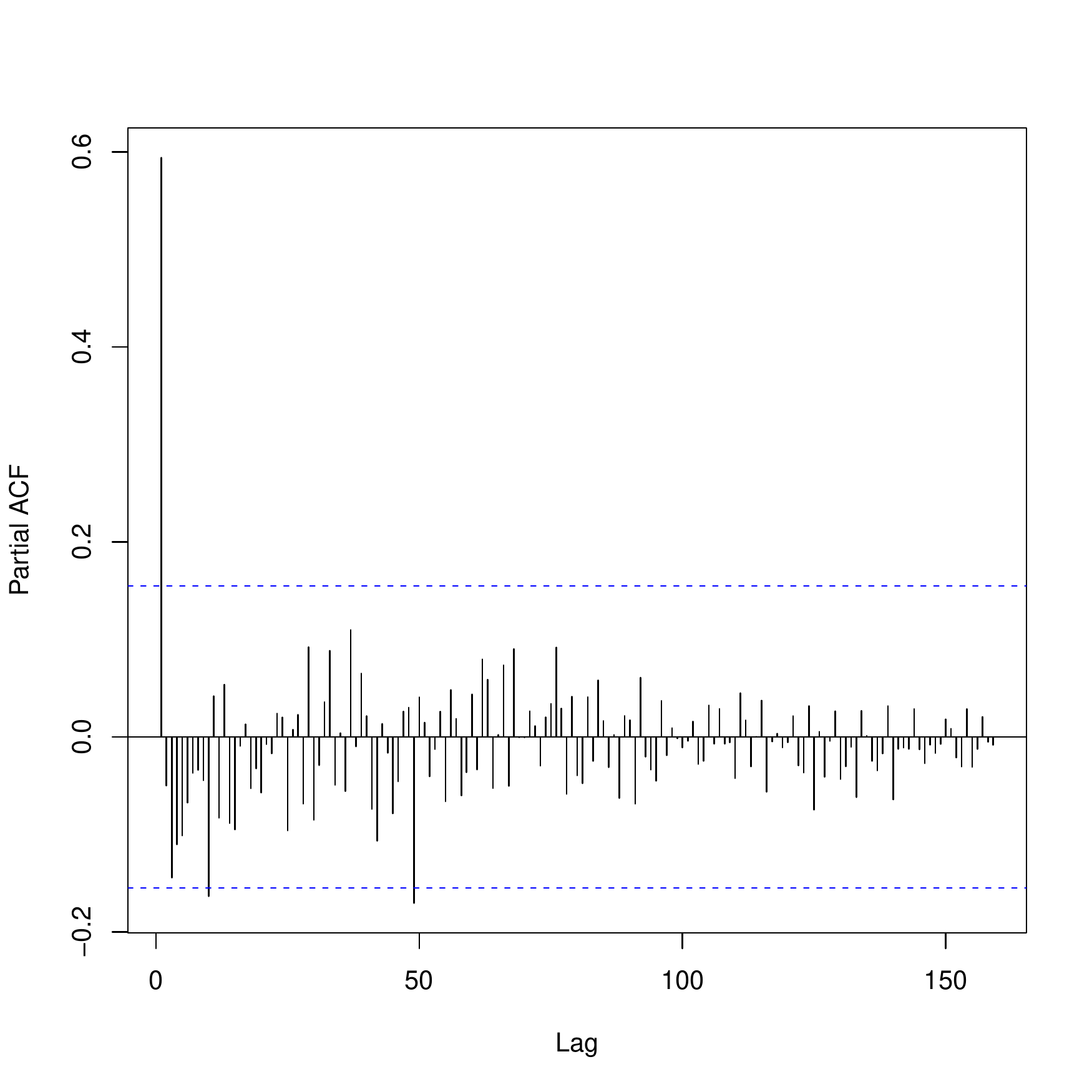}}
\end{subfloat}
\begin{subfloat}[ACF: \mbox{Across components}]
{\includegraphics[width=0.24\textwidth]{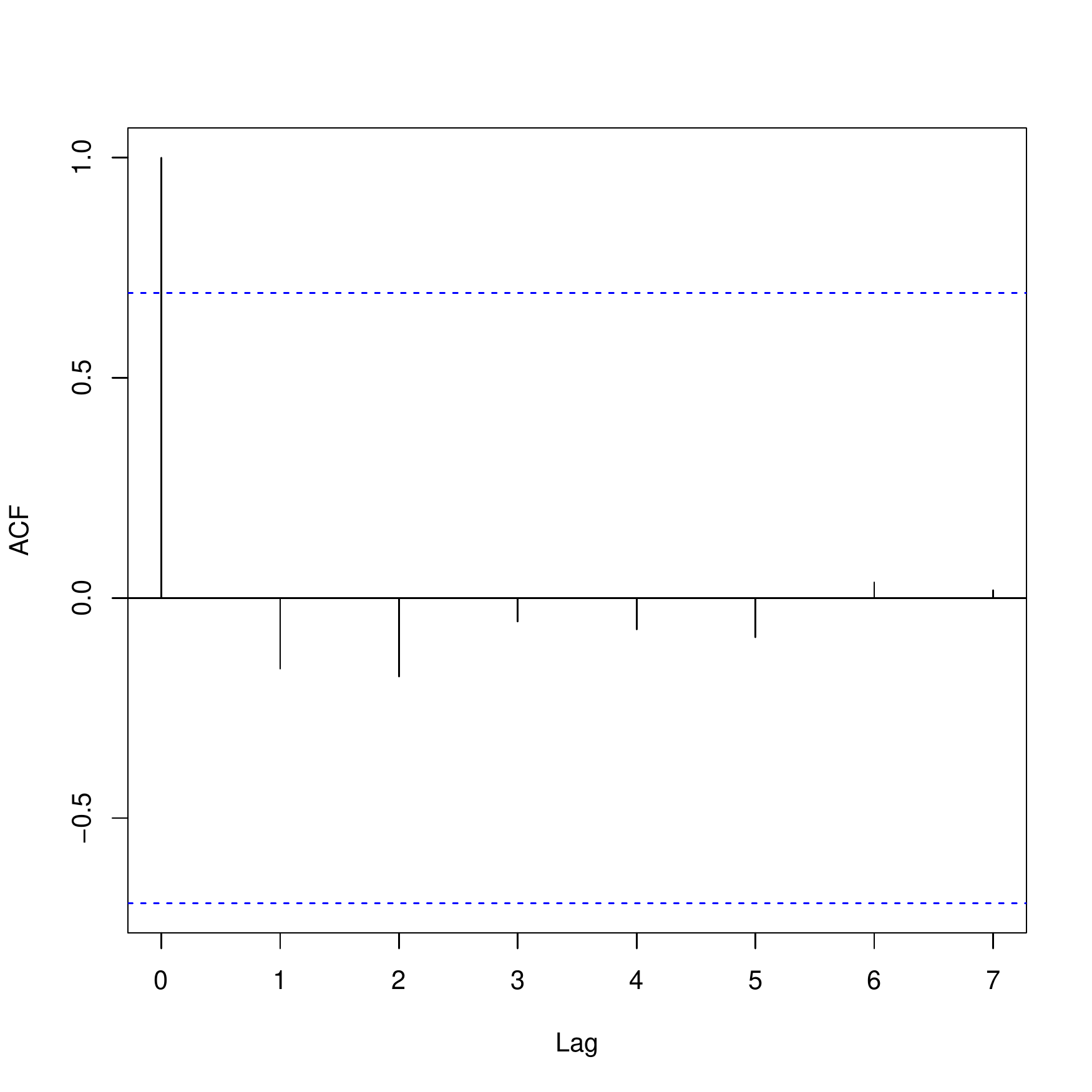}}
\end{subfloat}
\begin{subfloat}[PACF: \mbox{Across components}]
{\includegraphics[width=0.24\textwidth]{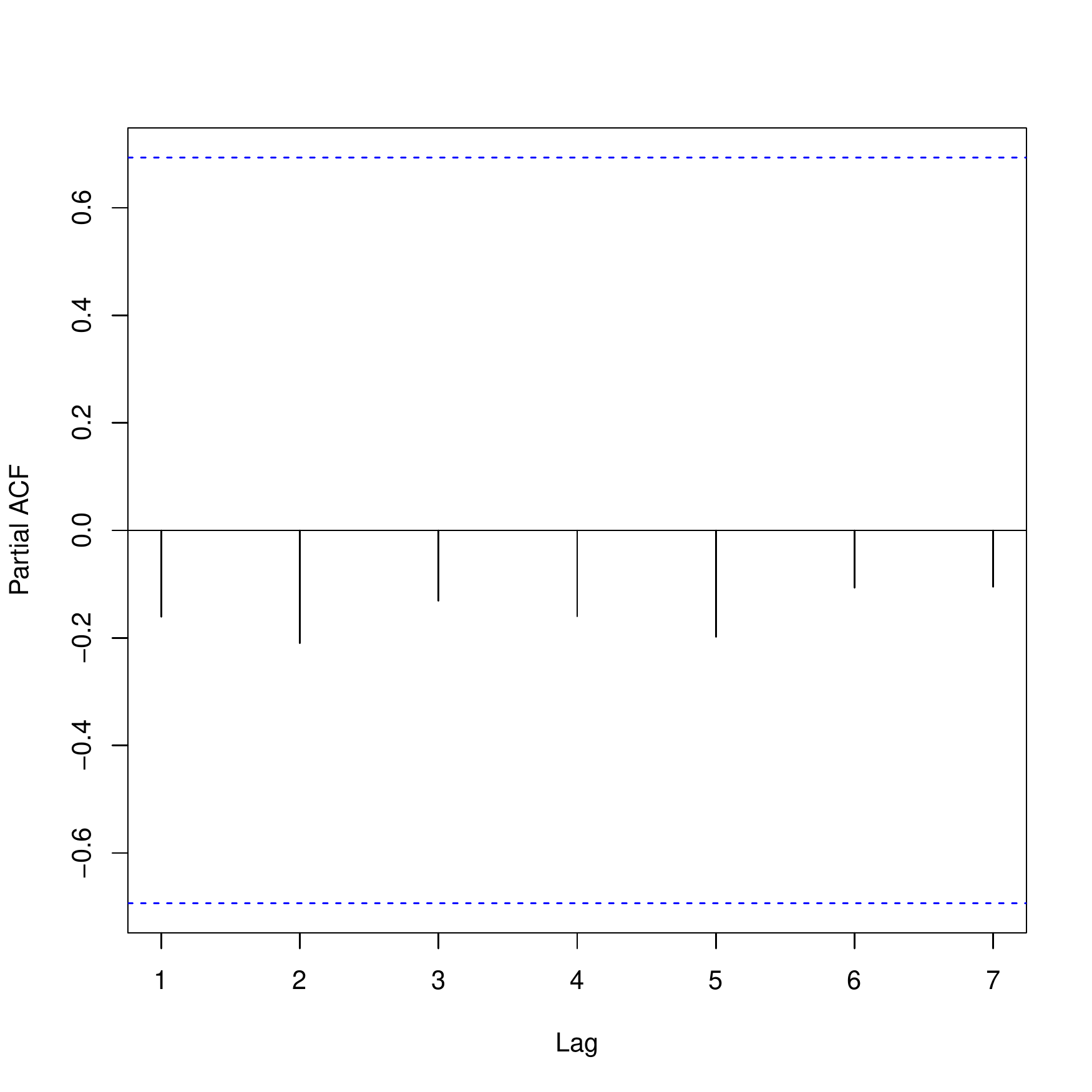}}
\end{subfloat}

\caption{ACF and PACF for the estimated errors  from the left trajectory}
\label{figure_ACF_1}
\end{figure}

We now return to the gas emission data example outlined in Section \ref{sec:intro}.
{\cred This data example has already been analyzed by \cite{Hirst2013} who adopt a Bayesian approach. In brief, they model atmospheric point concentration measurements as the sum of a spatially and temporally smooth atmospheric background, augmented by concentrations from local sources. 
Source emission rates are modelled by a Gaussian model taking possible multiple sources into account by means of a mixture model, whilst the atmospheric background concentration component is represented by a Markov random field. A reversible jump MCMC inference procedure is then used to provide point and uncertainty estimates for the plume origin. This approach also incorporates an optimisation approach to provide an initial point solution for inversion. These, and other necessary steps, combine to result in a computationally intensive procedure that relies on a multitude of parametric assumptions.  
In contrast, our approach makes fewer computational demands and requires only quite mild assumptions while still giving good results.  While we only analyse the case of a single source, our procedure can in principle be adapted to allow for multiple sources by appropriately defining change regions (which then no longer need to be epidemic). }

{\cred As pointed out in Section~\ref{sec:cpa:test} } a critical point for many change point tests and the corresponding estimators is the estimation of the long-run covariance matrix $\bs{\Sigma}$, which is a difficult problem statistically. There are two key aspects of the problem: first, time dependency and second, the large dimension of the covariance matrix with no structural assumption available. In the case of the gas emission data, time dependency is not negligible while the dependence between different components of the error process is very weak. By way of illustration, Figure  Figurs~\ref{figure_ACF_1} (a) and (b) show the empirical autocorrelation function (ACF) and partial autocorrelation function (PACF) for the first {\cred transect} of the estimated error sequence. These plots clearly indicate the presence of dependence. Equivalent analyses for the other components indicate similarly. See for example Chapter 18 in \cite{silke_diss}. In contrast, Figures~\ref{figure_ACF_1} (c) and (d) show the ACF and PACF for the estimated errors from one leg to the next for examplary time point~1, where no dependence is visible. Again, for other time points, a similar picture is obtained. To use ACF and PACF in the latter context makes sense keeping in mind that the original data was indeed a one-dimensional time series that has been transformed to a multivariate time series for the purpose of the data analysis. As such the vector of observations at each time point is indeed a thinned version of that original time series. This leads us to only estimate the long-run variances (i.e.\ the diagonal elements of $\bs{\Sigma}$) while setting the off-diagonal elements to zero. 


%
%

\begin{figure}[!b]
\begin{subfloat}[Estimated cloud based on the multivariate procedure]
	{\includegraphics[width=0.48\textwidth]{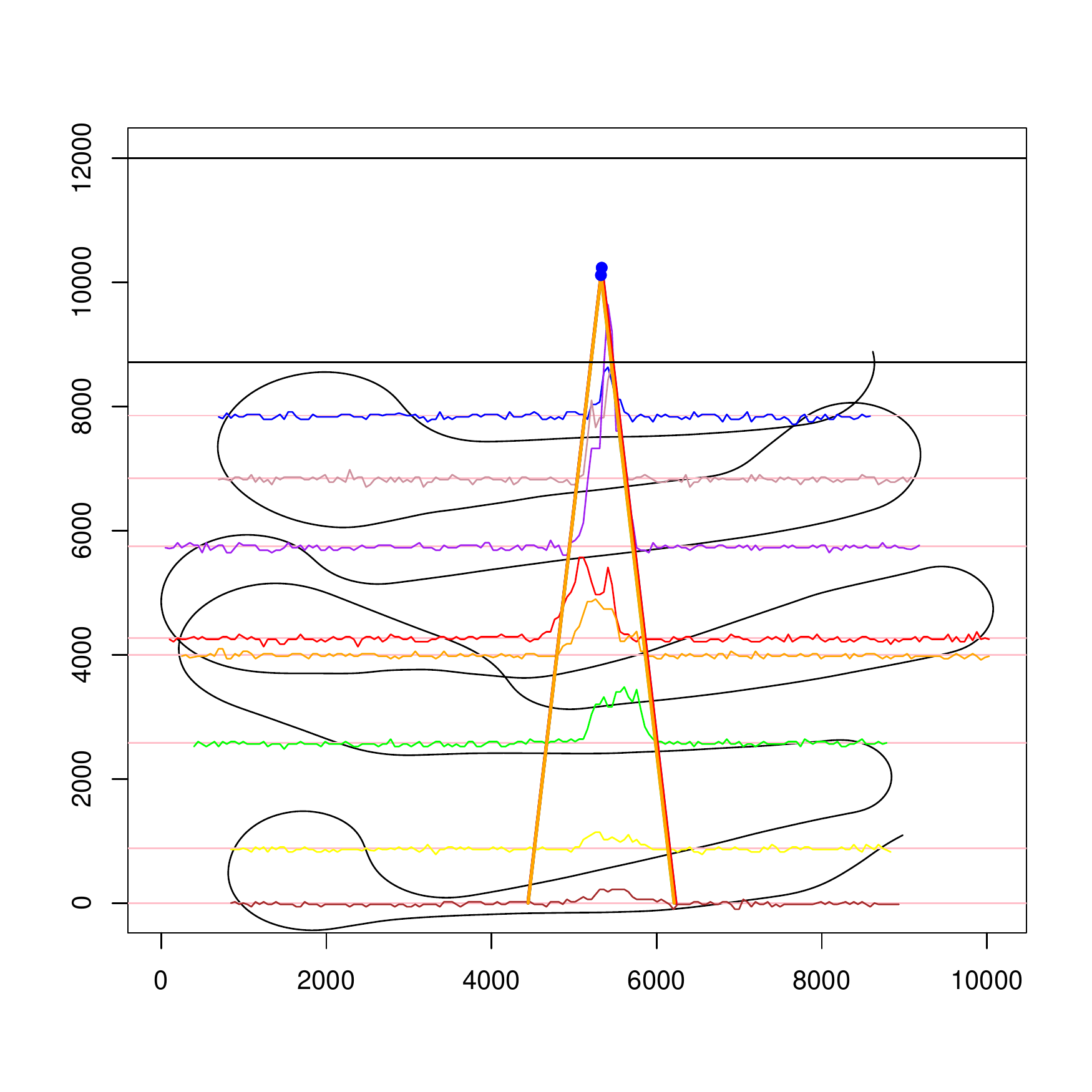}}
\end{subfloat}
\begin{subfloat}[Estimated cloud based on the projection procedure]
{\includegraphics[width=0.48\textwidth]{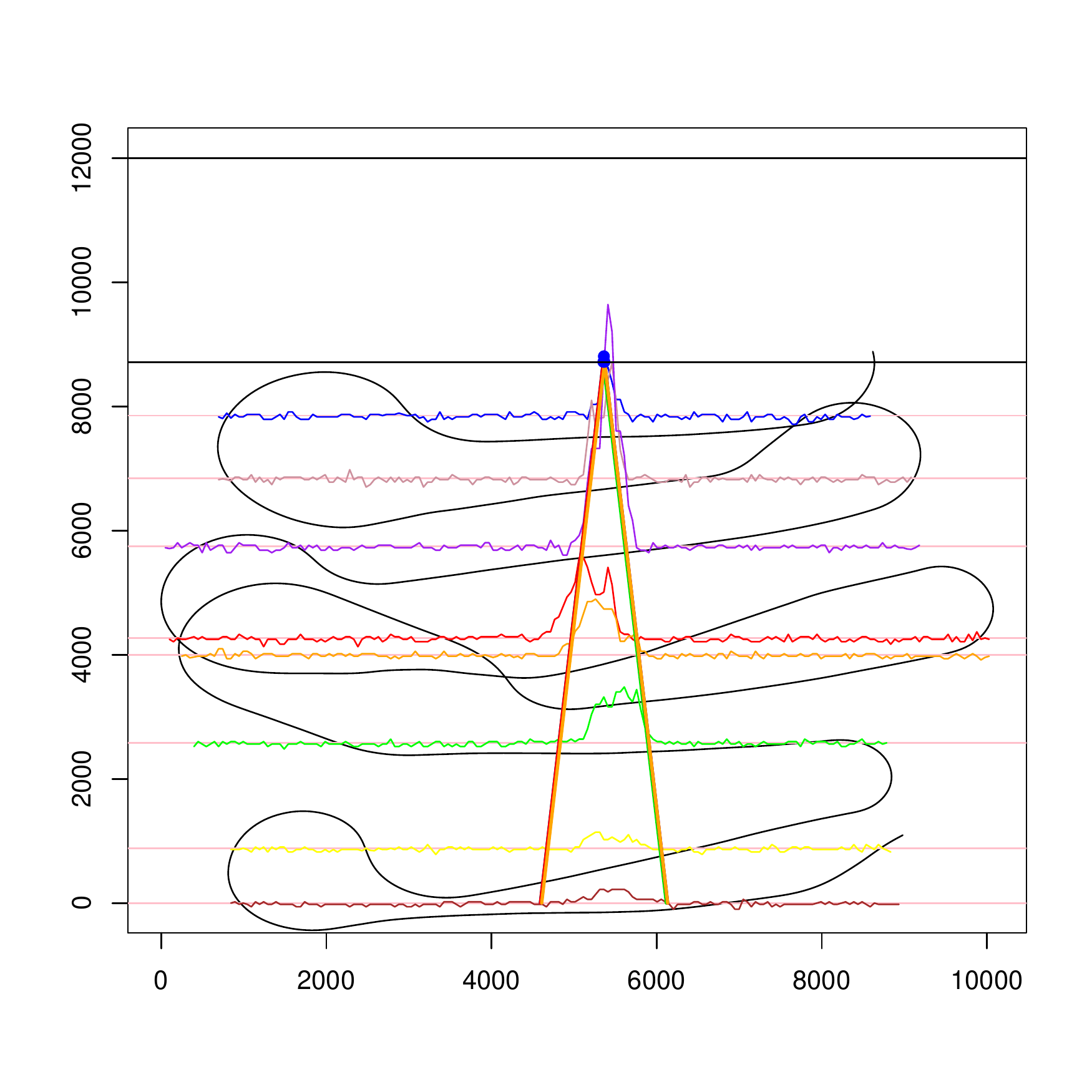}}
\end{subfloat}
\begin{subfloat}[Heatmap from the multivariate procedure]
{\includegraphics[width=0.98\textwidth]{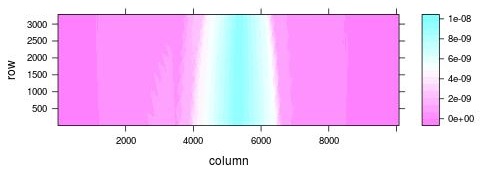}}
\end{subfloat}
\begin{subfloat}[Heatmap from the multivariate procedure]
{\includegraphics[width=0.98\textwidth]{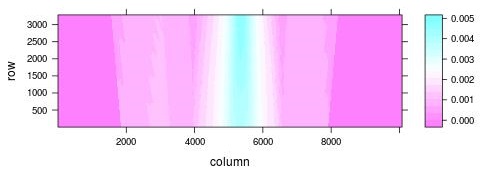}}
\end{subfloat}
\caption{Data Analysis for left trajectory}
\label{left_traj}
\end{figure}


In the following, we consider the left trajectory, while the analysis of the right trajectory is moved to the appendix (see Section~\ref{sec_right}). {\cred While the left trajectory can be considered well specified, the right trajectory is somewhat misspecified as the wind seems to have changed at some point. As such it gives some additional insight into the effect of misspecification.}
In both cases, we use a linear cloud as an approximation and do not have any knowledge about the actual opening angle. Thus we include a range of opening angles within the estimation procedure. While the simulations have shown that this can lead to some loss of precision, for the data analysis it does not seem to cause any problems. Figure~\ref{left_traj} (a) gives the corresponding cloud estimate for this data example and visual inspection suggests that a good fit has been obtained.  {\revtwo 
On the other hand, the figure also shows that the assumption of a constant mean within the change region is not met by the actual data at hand, where the concentration slowly increases before decreasing again. 
While the methodology of this paper could be adapted to this situation, this requires additional model assumptions on the shape of this gradual change. A substantial improvement of this approach can only be expected if such information is indeed available which is typically not the case, so that we decided once again to work with the simpler model.
}

Keeping in mind that the main objective is to get a good approximation of the source of the gas emission, it is also worthwhile considering a heatmap of the values of the statistic for each considered possible source location (where the maximum over all opening angles is given). This heat map for the left trajectory can be found in Figure~\ref{left_traj} (c). It becomes apparent that the statistic takes particularly high values in a vertical area in the middle, where differences are indeed very small. Effectively, all of those places can be considered possible source location so that this heat map can be used as a search map for the gas emission source. Furthermore, the reason why the values of the statistic are very close within that area  is the fact that a source closer to the lower end with a larger opening angle can approximate the signal as given by the discrete data set similarly well as a source closer to the upper end of the search area with a smaller opening angle. In a sense, this is related to an identifiability issue and what would be a flat likelihood surface in the context of maximum likelihood estimation.

For the projection estimator, we need to make additional assumptions on the decay of the gas intensity from one leg of the flight to the next. To this end, we use a function that first increase for the first legs before slowly decreasing to a similar level (for details please see  \cite[Figure 17.1]{silke_diss}). We first use an increasing level because the cloud is a 3D object, and the airplane flies at a certain height, so that the airplane first has to enter the cloud leading to first increasing levels before the dispersion effect of the gas leads to a slower decay again. This kind of behaviour can indeed be seen in both trajectories.
Figure ~\ref{left_traj} (b) shows the estimated cloud while Figure~\ref{left_traj} (d) shows the heat map. While the source of the cloud that has been picked by the projection estimator is different from the one picked by the multivariate estimator, the corresponding clouds do divide the time series in a very similar manner. This is a similar effect as has been described above in the context of the heat map related to weak identifiability. Considering the heat map of the projection statistic the area with high values of the statistics (that could be searched) is similar but smaller, which could indicate that the use of the additional information does indeed lead to a more precise estimation.

\section{Conclusions}\label{sec:conclusions}
The methodology developed in this work takes a  different view on multivariate change point analysis than the classical literature while  including those situations as a special case. In the setting we consider, the change points across components  no longer have to be aligned but can follow some kind of functional relationship. It is not necessary to know the functional relationship exactly but some reasonable parametrization needs to be available even if it depends on unknown parameter such as, for example, the precise shape of the cloud, its opening angle etc. 

The main contribution of the paper is the derivation of two different estimators for the unknown parameters of the functional relationship, at least some of which are the parameters of interest (such as the source of the cloud in the gas emission example). The first estimator only uses the parametric information of the functional relationship but allows for arbitrary change directions (as denoted by $\boldsymbol{\Delta}/\|\boldsymbol{\Delta}\|$ in this paper). As such it is related to classical estimators for multivariate change point situations with the difference that it is no longer the change points that are of interest but rather the underlying parameters of the functional parametrization of the changes. The second estimator relies on the additional knowledge of the change direction (not the magnitude of the change) and is related to classical change point estimators after an appropriate projection of the data into one dimension. This can greatly increase the precision of the estimators but at the risk of inconsistency or at less precision if that direction is not correct. Some simulations suggest that the procedure is not too sensitive with respect to mild deviations from the truth.

As a by product we obtain two testing procedures each related to one of the two estimators, for which we derive the limit distribution under the null as well as show consistency under alternatives. While these tests are not of immediate interest in the context of the gas emission example, they may be of independent interest in other situations.

For both estimators and both testing procedures, only very mild nonparametric assumptions on the error sequence are required and the case of dependent errors is also taken into account. We do not make any specific assumptions on this dependence but only need the validity of a functional central limit theorem which has been shown for many different dependent time series and weak dependency concepts.

The development of the methodology is motivated by an application of remote detection and location of the source of gas emissions based on aerial sensed-data and throughout the paper the development of the methodology has been explained by \mbox{means} of that data set, and of course finally analysed with the new methodology. While the methodology gives reasonable results, it can also be deduced that the exact source location is not strongly identifiable on the basis of this kind of data set.

Finally, all methods can be adapted to different but similar applications, situations or models, for example while an epidemic change setting is discussed in this paper, extensions to other scenarios are straightforward. Section~\ref{section_summary} explains the underlying ideas and construction principles to help with this task.

 \section*{Acknowledgements}
The authors are grateful to Bill Hirst and Philip Jonathan (Shell) for several valuable conversations that helped
motivate this work, and for providing access to the landfill data. {\cred The authors would also like to thank Philipp Klein (Otto-von-Guericke University, Magdeburg) for identifying a coding error in an earlier version  of this work.}

This work was supported by the grant 'Resampling procedures for high-dimensional change point tests of dependent data' financed by the state of Baden Wurttemberg. In
addition, support from the Karlsruhe House of Young Scientists (KHYS) for a research visit to Lancaster (UK),  the Isaac Newton Institute for Mathematical Sciences during the programme Statistical Scalability (supported by EPSRC grant numbers EP/K032208/1 and EP/R014604/1) and EPSRC (EP/N031938/1) are kindly acknowledged.

\bibliographystyle{plainnat}
\bibliography{IAErefs}

\newpage
\appendix
\part*{Appendix}
\section{Corresponding test procedures}\label{app_test}
{\cred In this section, we derive properties of the test statistics that are based on the data agglomeration techniques as discussed in Section~\ref{sec:cpa:test}.
First of all, } the multivariate statistic is defined as
 \begin{align*}
	 & 
	 T^M = \sup_{\vth\in\Theta} \frac{1}{N} A^M(\vth)= \sup_{\vth\in\Theta} \frac{1}{N} \bS_{\vth}^T \Sigma ^{-1} \bS_{\vth}.
 \end{align*}

 {\cred To derive the projection test statistic, first note that}
 the projected errors ${e}_P(t)$ are standardized if $\Sigma_A=\Sigma$ is correct. However, {\cred as already discussed this is usually too strong an assumption.}  Therefore, we stabilize the size of the test statistic with respect to possible misspecification (or misestimation) of $\Sigma$ by estimating the long-run variance of the projected~errors~by~$\widehat{\sigma}^2$. As a consequence, we obtain the projection statistic
 \begin{align*}
	 T^P &= \frac{1}{\sqrt{N}} \frac{1}{\widehat{\sigma}}  \sup_{\vth\in\Theta} A^P(\vth)=
	\frac{1}{\sqrt{N}} \frac{1}{\widehat{\sigma}}  \sup_{\vth\in\Theta}  \left| \sum\limits_{t=1}^{N} \left( \mathbf{D}_{\vth}(t/N) -\frac{1}{N} \sum\limits_{l=1}^{N} \mathbf{D}_{\vth}(l/N) \right) Y(t) \right|\\
	 &= \frac{1}{\sqrt{N}} \frac{1}{\widehat{\sigma}} \sup_{\vth\in\Theta}  \left| \sum\limits_{t=1}^{N}  \mathbf{D}_{\vth}(t/N) ( Y(t)-\bar{Y}_N) \right|,
\end{align*}
where $\mathbf{D}_{\vth}=\mathbf{D}_{\vth}^{\widetilde{\Delta},\widetilde{\Delta}}$.
 This statistic is related to the one by \citet{extreme_gradual} and \citet{limitdistr_gradual} that was obtained as the likelihood ratio statistic for a (non-epidemic) gradual change with a given polynomial slope.

 The following theorem establishes the null asymptotics of  these two {\cred test} statistics:
\begin{theorem}\label{theorem_null}
	Let $\{\boldsymbol{e}(\cdot)\}$ be a stationary time series that fulfils a functional central limit theorem towards a Wiener process with covariance matrix $\Sigma$.  Then, under the null hypothesis:
	\begin{enumerate}[(a)]
		\item For the multivariate statistic we obtain:
	\begin{align*}
		T^M\dto \sup_{\vth\in\Theta}\sum_{j=1}^d(B_j(G_{\vth}(j))-B_j(F_{\vth}(j))^2, 
	\end{align*}
where $\{B_j(\cdot)\}$, $j=1,\ldots,d$, are independent standard Brownian bridges. The assertion remain true, if $\Sigma$ is replaced by a consistent estimator $\widehat{\Sigma}$ (fulfilling $\widehat{\Sigma}\pto\Sigma$).
\item For the projection statistic we obtain (irrespective of the choice of $\widecheck{\boldsymbol{\Delta}}$ and $\Sigma$)	\begin{align*}
		&	T^P\dto\sup_{\vth\in\Theta} \left|\sum_{s\in \mathcal{M}_{\vth}} [\bD_{\vth}(s+)-\bD_{\vth}(s)] B(s)\right|,\\
		&\qquad \mathcal{M}_{\vth}=\{0<s<1: \bD_{\vth}(s+)-\bD_{\vth}(s) \neq 0\},
	\end{align*}
	where $\{B(\cdot)\}$ is a Brownian bridge,
	if $\widehat{\sigma}^2\pto\sigma_P^2$ with $\sigma_P^2=\sum_{h\in\ZZ}\cov(e^P(0),e^P(h))$. The assertion remains true if $\Sigma_A$ is replaced by $\widehat{\Sigma}_N$ with $\widehat{\Sigma}_N\pto \Sigma_A$. 	
	\end{enumerate}
\end{theorem}
In the second part of the theorem it is important to note that the gradual change in our example is in fact step-wise constant with discontinuity points in $\mathcal{M}$. If instead a slope is assumed that is differentiable, then one gets an integral of the Brownian bridge weighted by the derivative of the slope as a limit \citep{limitdistr_gradual}.

The assumption on the error sequence is very weak. For independent errors with second moments it follows from the famous Donsker theorem (Theorem 16.1 in \citet{billingsley}). Subsequently, it has been proven for many different types of weak dependence, including (but not limited to) mixing or $L^4$-approximation (see for example \citet{fclt_mixing}, \citet{fclt_weak_mixing} or \citet[Appendix A]{Aue2009}).
\begin{rem}[Misspecification of the covariance matrix]\label{rem_mis_cov}
	If $d$ is relatively large, then ${\Sigma}$ cannot be estimated well without making further assumptions (such as diagonality or sparsity). For this reason, it is also of interest to understand the behavior of the statistics under misspecification, i.e.\ if $\widehat{\Sigma}_N\to\Sigma_A$ for some positive definite matrix $\Sigma_A$. The projection statistic is robust in this respect under the null hypothesis, in the sense that the same limit applies, because we can easily estimate the (long-run-)variance of the projected errors consistently. This is not true for the multivariate statistic for which the Brownian bridges $\left\{B_1(\cdot),\ldots,B_d(\cdot)\right\}$ in the limit distribution are no longer independent but have the covariance matrix $\Sigma_A^{-1/2}\Sigma^{1/2}$. In this case bootstrap methods can help (see \citet{Aston2012}).
	\end{rem}

	\begin{rem}[Weighted versions of the statistics]\label{rem_weighted}
		{\cred For a classical change point tests, weighted versions of the statistics can help increase power of the test if the change occurs at certain time points for example as a way of incorporating a-priori information about the location of the change point (see e.g.~\cite{eeg_data}). Similarly, we can use weight functions here to increase power for certain source locations.}		\begin{enumerate}[(a)]
			\item \textbf{Multivariate Statistics}\\In our setup we can, e.g., use
		$$T^M(w)=\sup_{\vth\in\Theta}w^2_M(\vth) \bS_{\vth}^T\Sigma^{-1}\bS_{\vth}\dto \sup_{\vth} w_M(\vth)^2\sum_{j=1}^d(B_j(G_{\vth}(j))-B_j(F_{\vth}(j))^2 $$
		with a weight function $w_M$ that fulfills $\sup_{\vth\in\Theta}w^2_M(\vth)<\infty$. Alternatively, for a diagonal matrix $\Sigma$, we can penalize each component separately and consider 
		$$\widetilde{T}^M(w)=\sup_{\vth\in\Theta} \bS^{wT}_{\vth}\Sigma^{-1}\bS^w_{\vth}\dto\sup_{\vth} \sum_{j=1}^dw_M(\vth,j)^2(B_j(G_{\vth}(j))-B_j(F_{\vth}(j))^2 $$
		with $S^w_\vth(i)=w_M(\vth,i) S_{\vth}(i)$ as long as $\sup_{i=1,\ldots,d}\sup_{\vth\in\Theta}w_M^2(\vth,i)<\infty$. In the simulation study in \citet{silke_diss} the latter approach is implemented with $w_M(\vth,i)=(G_{\vth}(i)-F_{\vth}(i))^{-\beta}\,(1-G_{\vth}(i)+F_{\vth}(i))^{-\beta}$, $0\ls \beta\ls \frac 1 2$, which fulfills the above assumption if $\epsilon\ls G_{\vth}(i)-F_{\vth}(i)\ls 1-\epsilon$ for all $i=1\ldots,d$ and $\vth\in\Theta$ for some $\epsilon>0$. In  gas emission example this assumption translates to assumptions on the minimum and maximum opening angle as well as the minimal and maximal distance of a possible source from the first and last leg of the flight.
		This is a typical weight function in the classical model with $F_{\vth}(i)=\lambda_1$ and $G_{\vth}(i)=\lambda_2$.
\item \textbf{Projection Statistics}\\	Due to the gradual change a different type of weight function is necessary for the projection statistic, namely 
		\begin{align*}
			w_P(\vth)=\left(\frac 1 N\sum_{j=1}^N\left(\bD_{\vth}(j/N)-\frac{1}{N}\sum_{i=1}^N\bD_{\vth}(i/N)\right)^2\right)^{-\beta},\quad 0\ls \beta\ls \frac 1 2.
		\end{align*}
		Then it holds 
		\begin{align*}
			&\frac{1}{\sqrt{N}} \frac{1}{\widehat{\sigma}} \sup_{\vth\in\Theta}w_P(\vth)  \left| \sum\limits_{t=1}^{N}  \mathbf{D}_{\vth}(t/N) ( Y(t)-\bar{Y}_N) \right|
			\dto \sup_{\vth\in\Theta} \frac{\left|\sum_{s\in \mathcal{M}_{\vth}} [\bD_{\vth}(s+)-\bD_{\vth}(s)] B(s)\right|}{\left(\int_0^1(\bD_{\vth}(z)-\int_0^1 \bD_{\vth}(y)\,dy)^2\,dz\right)^{\beta}},
		\end{align*}
		if $\sup_{\vth\in\Theta}\left(\int_0^1(\bD_{\vth}(z)-\int_0^1 \bD_{\vth}(y)\,dy)^2\,dz\right)^{-\beta}<\infty$, which again translates to assumptions on the minimal and maximal opening angles as well as minimal and maximal distances to the first and last leg of the flight.

		For estimation purposes based on the projection statistic we need to use the above statistic with $\beta=\frac 1 2$ in order to obtain consistent results.
\end{enumerate}
	\end{rem}

Under the alternative the procedures have asymptotic power one as suggested by the next theorem:

\begin{theorem}\label{theorem_alt}Let the assumptions on the errors of Theorem~\ref{theorem_null} hold. Then, under a fixed alternative as in \eqref{eq_model}  with $\Delta_i\neq 0$ for at least one $i=1,\ldots,d$ it holds:
	\begin{enumerate}[(a)]
		\item For the multivariate statistic it holds $T^M\pto \infty$, i.e.\ it has asymptotic power one.  This remains true, if an estimator for $\Sigma$ is used as long as $\widehat{\Sigma}\pto \Sigma_A$ for some positive definite $\Sigma_A$ (which can be different from $\Sigma$).
		\item For the projection statistic with correct projection direction it holds  $T^P\pto \infty$, i.e.\ it has asymptotic power one, as long as $\widehat{\sigma}^2\pto\sigma^2_A\neq 0$ and  $\mathbf{D}_{\vth_0}$ is not constant.	\end{enumerate}
\end{theorem}

The assumption on the signal function $\mathbf{D}_{\vth}$ holds true for the linear cloud that has been used in the data analysis section.

\begin{rem}\label{rem_alt}
	The projection test remains consistent for a misspecified projection direction $\widecheck{\Delta}$ (where  $\widetilde{\Delta}$ or $\Sigma_A$ can be misspecified) as long as there exists $\vth_1$ with
\begin{align*}
	\int_0^1\mathbf{D}_{\vth_1}(s)\left( \mathbf{D}_{\vth_0}^{\Delta,\widetilde{\Delta}}(s)-\int_0^1\mathbf{D}_{\vth_0}^{\Delta,\widetilde{\Delta}}(z)\,dz \right)\,ds\neq 0.
\end{align*}
The same holds true if the shape of the cloud is misspecified if the above assertion holds with $\mathbf{D}_{\vth_0}^{\Delta,\widetilde{\Delta}}$ being the projected signal of the true shape function.

Similarly, the multivariate test statistic is consistent for a misspecified cloud  under weak conditions (see Remark 15.2 in \citet{silke_diss}).
\end{rem}

The above assertions also remain true for weighted versions of the corresponding test statistics (see Remarks 15.1 and 15.3 in \citet{silke_diss}).


\section{Proofs}\label{sec:proofs}
In this section, we give the proofs of the previous theorems. More detailed versions of the proofs can be found in \cite{silke_diss}.

\begin{proof}[of Theorem~\ref{theorem_null}]
	The assertion of (a) follows immediately from the functional central limit theorem on noting that the statistic is a continuous functional of the partial sum process. For details we refer to the proof of Theorem~14.1 in \cite{silke_diss}.

	With summation by parts 
	we derive the equality
\begin{align*}
	&\frac{1}{\sqrt{N}} \sum\limits_{i=1}^{N} \bD_{\vth}(i/N)\left(\frac{i}{N}\right) \left( Y(i) - \overline{Y}_N  \right)\\
&= \bD_{\vartheta}(1) \frac{1}{\sqrt{N}} \sum\limits_{i=1}^{N} \left( Y(i) - \overline{Y}_N  \right) \\
&\quad - \sum\limits_{i=1}^{N-1} \left( \bD_{\vartheta}\left( \frac{i+1}{N} \right) - \bD_{\vartheta}\left( \frac{i}{N} \right) \right) \frac{1}{\sqrt{N}} \sum\limits_{j=1}^{i} \left( Y(j) - \overline{Y}_N \right)\\
&= - \sum_{s \in \mathcal{M}_{\vartheta}} \left( \bD_{\vartheta}(s+) - \bD_{\vartheta}(s) \right) \frac{1}{\sqrt{N}} \sum\limits_{j=1}^{\lfloor Ns \rfloor} \left( Y(j) - \overline{Y}_N \right),
\end{align*}
 where the last line follows because $\bD_{\vth}(\cdot)$ is piecewise constant and has at most $2d$ points of discontinuity given in $\mathcal{M}_{\vth}$.
Under the null hypothesis it holds
\begin{align*}
&\frac{1}{\sqrt{N}} \sum\limits_{j=1}^{\lfloor Ns \rfloor} \left( Y(j) - \overline{Y}_N \right)
= \frac{1}{\sqrt{N}} \sum\limits_{j=1}^{\lfloor Ns \rfloor} \left( e_P(j) - \overline{e}_{P,N} \right)\\
&= \frac{1}{\sqrt{N}} \left( \sum\limits_{j=1}^{\lfloor Ns \rfloor} \left( \boldsymbol{e}(j) - \overline{\boldsymbol{e}}_N \right) \right)^T {\Sigma}_A^{-1} \widetilde{\Delta}/{\|\Sigma_A^{-1/2}\widetilde{\boldsymbol{\Delta}}\|}.
\end{align*}
We can now conclude the assertion from the functional central limit theorem of the error terms in addition to an application of the continuous mapping theorem  on noting that the variance of the above term is in fact $s\,\var(e_P(1))$.

Standard arguments yield the assertion in the case, where $\Sigma_A$ is consistently estimated.
\end{proof}

\begin{proof}[of the Remarks~\ref{rem_mis_cov} and \ref{rem_weighted}]
	The assertions of Remarks~\ref{rem_mis_cov} as well as ~\ref{rem_weighted}(a) can be obtained analogously to the above proof. For (b) we need to notice that due to the at most $2d$ discontinuity points in addition to $\sup_{\vth}\sup_s\bD_{\vth}(s)<\infty$, $w_P(\vth)$ converges uniformly to the $ \left(\int_0^1(\bD_{\vth}(z)-\int_0^1 \bD_{\vth}(y)\,dy)^2\,dz\right)^{-\beta}$. Then we can conclude as before.
\end{proof}

For the proof of the two results under the alternative, we need the following auxiliary lemma.
\pagebreak
\begin{lemma}\label{lemma_alt}Let the assumptions of Theorem~\ref{theorem_alt} hold. 
	\begin{enumerate}[(a)]
		\item The multivariate statistic yields the following signal:
\begin{align*}
	&\sup_{i=1, \cdots, d} \sup_{\vartheta \in \Theta} \left| \frac{1}{N} \sum_{t=\lfloor N{F}_{\vartheta}(i)\rfloor +1}^{\lfloor N {G}_{\vartheta}(i)\rfloor} \left( X_i(t) - \frac{1}{N} \sum_{l=1}^{N} X_i(l) \right) -\Delta_i\, h_{\vth}(i)\right|=o_P(1),
\end{align*}
where $h_{\vth}(i)=g_{F_{\vartheta_0}(i),G_{\vartheta_0}(i)}({G}_{\vartheta}(i)) - g_{F_{\vartheta_0}(i),G_{\vartheta_0}(i)}({F}_{\vartheta}(i))$ with 
\begin{align*}
g_{t_0,t_1}(s) = \begin{cases}
-s (t_1 - t_0), & s \leq t_0,\\
s (1 - (t_1 - t_0)) - t_0&  t_0 < s \leq t_1,\\
(1-s) (t_1 - t_0), & s > t_1.
\end{cases}
\end{align*}
In particular: $h_{\vth_0}(i)=(G_{\vth_0}(i)-F_{\vth_0}(i))(1-(G_{\vth_0}(i)-F_{\vth_0}(i)))>0.$
\item The projection statistic yields the following signal:
	\begin{align*}
		\sup_{\vth\in\Theta}  \left| \frac{1}{N}\sum\limits_{t=1}^{N} \mathbf{D}_{\vth}(t/N) \left(Y(t)-\bar{Y}_N\right)-\int_0^1\mathbf{D}_{\vth}(s)\left( \mathbf{D}_{\vth_0}^{\Delta,\widetilde{\Delta}}(s)-\int_0^1\mathbf{D}_{\vth_0}^{\Delta,\widetilde{\Delta}}(z)\,dz \right)\,ds \right|=o_P(1).
	\end{align*}
	The assertion remains true if $\Sigma_A$ is consistently estimated.
\end{enumerate}
	\end{lemma}

	\begin{proof}
Some elementary calculations yield%
\begin{align*}
	& \frac{1}{N} \sum_{t=\lfloor N{F}_{\vartheta}(i)\rfloor +1}^{\lfloor N {G}_{\vartheta}(i)\rfloor} \left( X_i(t) - \frac{1}{N} \sum_{l=1}^{N} X_i(l) \right) -\Delta_i\, h_{\vth}(i)
\\
&= \frac{1}{N} \sum_{t=\lfloor N{F}_{\vartheta}(i)\rfloor +1}^{\lfloor N {G}_{\vartheta}(i)\rfloor} \left( e_i(t) - \frac{1}{N} \sum_{l=1}^{N} e_i(l) \right) +o(1)
\end{align*}
uniformly in $i$ and $\vth$. Assertion (a) now follows from an application of the (multivariate) functional central limit theorem for the error terms $\{\mathbf{e}(\cdot)\}$.

For (b) first note
\begin{align*}
	&\sup_{\vth\in\Theta}  \left| \frac{1}{N}\sum\limits_{t=1}^{N} \mathbf{D}_{\vth}\left(\frac t N\right) \left(Y(t)-\bar{Y}_N\right)-\int_0^1\mathbf{D}_{\vth}(s)\left( \mathbf{D}_{\vth_0}^{\Delta,\widetilde{\Delta}}(s)-\int_0^1\mathbf{D}_{\vth_0}^{\Delta,\widetilde{\Delta}}(z)\,dz \right)\,ds \right|\\
	&\leq \sup_{\vartheta \in \Theta} \left| \frac{1}{N} \sum_{t=1}^{N} \mathbf{D}_{\vartheta}\left(\frac tN\right) \left( \mathbf{D}_{{\vartheta_0}}^{\Delta,\widetilde{\Delta}}\left(\frac{t}{N}\right) - \overline{\mathbf{D}}^{\Delta,\widetilde{\Delta}}_{{\vartheta_0}} \right) - \int_0^1\mathbf{D}_{\vth}(s)\left( \mathbf{D}_{\vth_0}^{\Delta,\widetilde{\Delta}}(s)-\int_0^1\mathbf{D}_{\vth_0}^{\Delta,\widetilde{\Delta}}(z)\,dz \right)\,ds  \right|\\
&\quad + \sup_{\vartheta \in \Theta} \left| \frac{1}{N} \sum_{t=1}^{N} D_{{\vartheta}}\left(\frac tN\right) \left( e_P(t) - \overline{e}_P \right) \right|=o_P(1),
\end{align*} 
where the assertion for the first summand follows from standard arguments, while the assertion for the second summand follows from Theorem~\ref{theorem_null}. If a consistent estimator $\widehat{\Sigma}_N$ for $\Sigma_A$ is used, the assertion for the first summand follows similarly, while the assertion for the second summand follows from Remark~\ref{rem_mis_cov}.

		\end{proof}

		\begin{proof}[of Theorem~\ref{theorem_alt} and Remark~\ref{rem_alt}]
		By Lemma~\ref{lemma_alt}(a) it holds
		\begin{align*}
			&T^M \gs  \frac{1}{N}  \bS_{\vth_0}^T \widehat{\Sigma}^{-1} \bS_{\vth_0} =
			N\, \left(\boldsymbol{H}_{\vth_0}^T\Sigma_A^{-1}\boldsymbol{H}_{\vth_0}+o_P(1)  \right)\pto \infty,\\
			&\text{where}\quad \boldsymbol{H}_{\vth_0}^T=(H_{\vth_0}(1)^,\ldots,H_{\vth_0}(d))\neq \mathbf{0}\text{ as } H_{\vth_0}(i)=\Delta_i\,h_{\vth_0}(i).
		\end{align*}
			This completes the proof of (a).

			For (b) we obtain similarly by Lemma~\ref{lemma_alt}(b)	
			\begin{align*}
				&T^P\gs \sqrt{N}\left( \left| \int_0^1\mathbf{D}_{\vth_1}(s)\left( \mathbf{D}_{\vth_0}^{\Delta,\widetilde{\Delta}}(s)-\int_0^1\mathbf{D}_{\vth_0}^{\Delta,\widetilde{\Delta}}(z)\,dz \right)\,ds\right|+o_P(1) \right).
			\end{align*}
			This completes the proof of Remark~\ref{rem_alt}. If $\widetilde{\Delta}=\Delta/c$ for some constant $c>0$, then by an application of Jenssens inequality and the fact that equality only holds in Jenssens equality for constant functions, we have indeed (with $\vth_1=\vth_0$)
			\begin{align*}
				\int_0^1\mathbf{D}_{\vth_0}(s)\left( \mathbf{D}_{\vth_0}^{c\widetilde{\Delta},\widetilde{\Delta}}(s)-\int_0^1\mathbf{D}_{\vth_0}^{c\widetilde{\Delta},\widetilde{\Delta}}(z)\,dz \right)\,ds\neq 0.
			\end{align*}
		\end{proof}

	\begin{proof}[of Theorem~\ref{theorem_est} and Remark~\ref{rem_est}]
		With $\mathbf{H}_{\vth}=(H_{\vth}(1),\ldots,H_{\vth}(d))^T$ with $H_{\vth}(i)=\Delta_i \,h_{\vth,\vth_0}(i)$  as in Lemma~\ref{lemma_alt} it holds 
\begin{align*}
&\sup_{\vartheta \in \Theta} \left| \left( \frac {1}{N^2} \bs{S}_{\vartheta}^T \widehat{\bs{\Sigma}}_N^{-1} \bs{S}_{\vartheta} \right)^{\frac 12} - \left( \bs{H}_{\vartheta}^T \bs{\Sigma}_A^{-1}\bs{H}_{\vartheta} \right)^{\frac 12} \right|\\
&\leq \sup_{\vartheta \in \Theta}\left( \left\|  \widehat{\bs{\Sigma}}_N^{-\frac{1}{2}} \left( \frac {1}{N} \bs{S}_{\vartheta}
- \bs{H}_{\vartheta} \right) \right\|
+ \left\| \left( \widehat{\bs{\Sigma}}_N^{-\frac{1}{2}}
 - \bs{\Sigma}_A^{-\frac{1}{2}} \right) \bs{H}_{\vartheta} \right\|\right)=o_P(1).
\end{align*}

Because $\vth\mapsto \|\Sigma_A^{-1/2}\,\mathbf{H}_{\vth}\|$ is continuous, standard arguments show that the estimator from the multivariate procedure in (a) converges to the maximizer of the signal part $\|\Sigma_A^{-1/2}\,\mathbf{H}_{\vth}\|$ if this maximizer is unique. Consequently, the assertion of Remark~\ref{rem_est} follows. We will now show that for a diagonal matrix $\Sigma_A$ the maximizer is uniquely given by $\vth_0$. Because the function $g_{t_0,t_1}(s), \ s \in [0,1], 0 \leq t_0<t_1 \leq 1$, defined in Lemma \ref{lemma_alt}, is piecewise constant with a unique minimum at $t_0$ and  a  unique maximum at $t_1$, the difference $g_{t_0,t_1}(s_1)-g_{t_0,t_1}(s_0)$ has a unique maximum in $(s_0,s_1)=(t_0,t_1)$. Consequently, $h_{\vth}(i)$ has a maximum in $\vth_0$ for each $i$, in  particular,
$$| H_{\vartheta}(i) | \leq | H_{\vartheta_0}(i) |, \quad \forall \ \vartheta \in \Theta, \ i=1, \dots d.$$
Furthermore, due to the identifiability uniqueness condition there exists an $i=1,\ldots,d$ for each $\vth\neq\vth_0$ such that
\begin{equation*}
\ | H_{\vartheta}(i) | < | H_{\vartheta_0}(i) |.
\end{equation*}
Since $\boldsymbol{\Sigma}_A=\mbox{diag}(s_1^2, \dots, s_d^2), \ s_i>0$ is a diagonal matrix, we finally get for all $\vth\neq\vth_0$
\begin{align*}
& \left\| \boldsymbol{\Sigma}_A^{-\frac{1}{2}} \boldsymbol{H}_{\vartheta} \right\|^2
= \sum_{i=1}^d \frac 1{s_i^2} H_{\vartheta}^2(i)
< \sum_{i=1}^d \frac 1{s_i^2} H_{\vartheta_0}^2(i)
= \left\| \boldsymbol{\Sigma}_A^{-\frac{1}{2}} \boldsymbol{H}_{\vartheta_0} \right\|^2,
\end{align*}
completing the proof of (a).

For the proof of (b) we first obtain by Lemma~\ref{lemma_alt} (b) and some elementary arguments
\begin{align*}
	&	\frac{\frac 1 N \sum\limits_{t=1}^{N}  \mathbf{D}_{\vth}(t/N) ( Y(t)-\bar{Y}_N) }{\left(\frac 1 N\sum_{t=1}^N\left(\mathbf{D}_{\vth}(t/N) -\frac 1 N\sum_{l=1}^N\mathbf{D}_{\vth}(l/N)\right)^2  \right)^{1/2}}\\
	&=\frac{\int_0^1\mathbf{D}_{\vth}(s)\left( \mathbf{D}_{\vth_0}^{\Delta,\widetilde{\Delta}}(s)-\int_0^1\mathbf{D}_{\vth_0}^{\Delta,\widetilde{\Delta}}(z)\,dz \right)\,ds}{\left( \int_0^1\left( \mathbf{D}_{\vth}(s)-\int_0^1\mathbf{D}_{\vth}(z)\right)^2\,dz \right)^{1/2}}+o_P(1)\\
	&=\int_0^1\widetilde{\mathbf{D}}_{\vth}(s)\,\left( \mathbf{D}_{\vth_0}^{\Delta,\widetilde{\Delta}}(s)-\int_0^1\mathbf{D}_{\vth_0}^{\Delta,\widetilde{\Delta}}(z)\,dz \right)\,ds+o_P(1),
\end{align*}
with the notation of Remark~\ref{rem_alt} (b). Obviously, the signal part is maximized iff 
$	\int_0^1\widetilde{\mathbf{D}}_{\vth}(s)\,\widetilde{\mathbf{D}}_{\vth_0}^{\Delta,\widetilde{\Delta}}(s)\,ds$
is maximized. Standard arguments give the assertion of the Remark. For assertion (b) of Theorem~\ref{theorem_est} we need to show that $\vth_0$ is the unique maximizer of this expression under the given assumptions. Indeed by the Cauchy-Schwarz inequality it holds
\begin{align*}
	\int_0^1\widetilde{\mathbf{D}}_{\vth}(s)\,\widetilde{\mathbf{D}}_{\vth_0}^{\Delta,\widetilde{\Delta}}(s)\,ds\ls 1,
\end{align*}
where equality holds iff $\widetilde{\mathbf{D}}_{\vth}=c\,\widetilde{\mathbf{D}}_{\vth_0}^{\Delta,\widetilde{\Delta}}$ for some constant $c$. Because of the given identifiability uniqueness condition, this only holds for $\vth=\vth_0$, completing the proof.
	\end{proof}

	\section{Data analysis of the right trajectory}\label{sec_right}
	
\begin{figure}[tb]
\begin{subfloat}[Estimated cloud based on the multivariate procedure]
	{\includegraphics[width=0.48\textwidth]{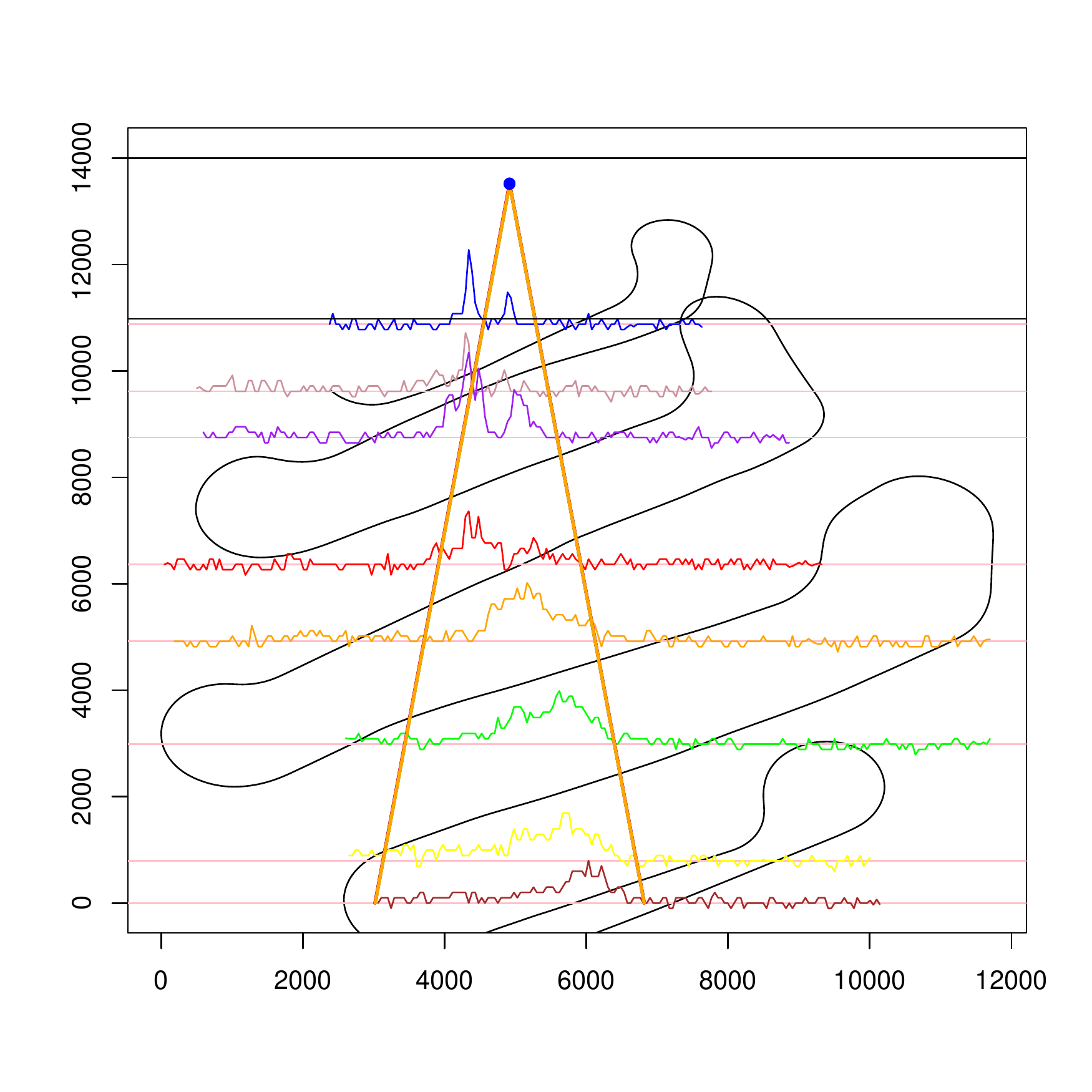}}
\end{subfloat}
\begin{subfloat}[Estimated cloud based on the projection procedure]
{\includegraphics[width=0.48\textwidth]{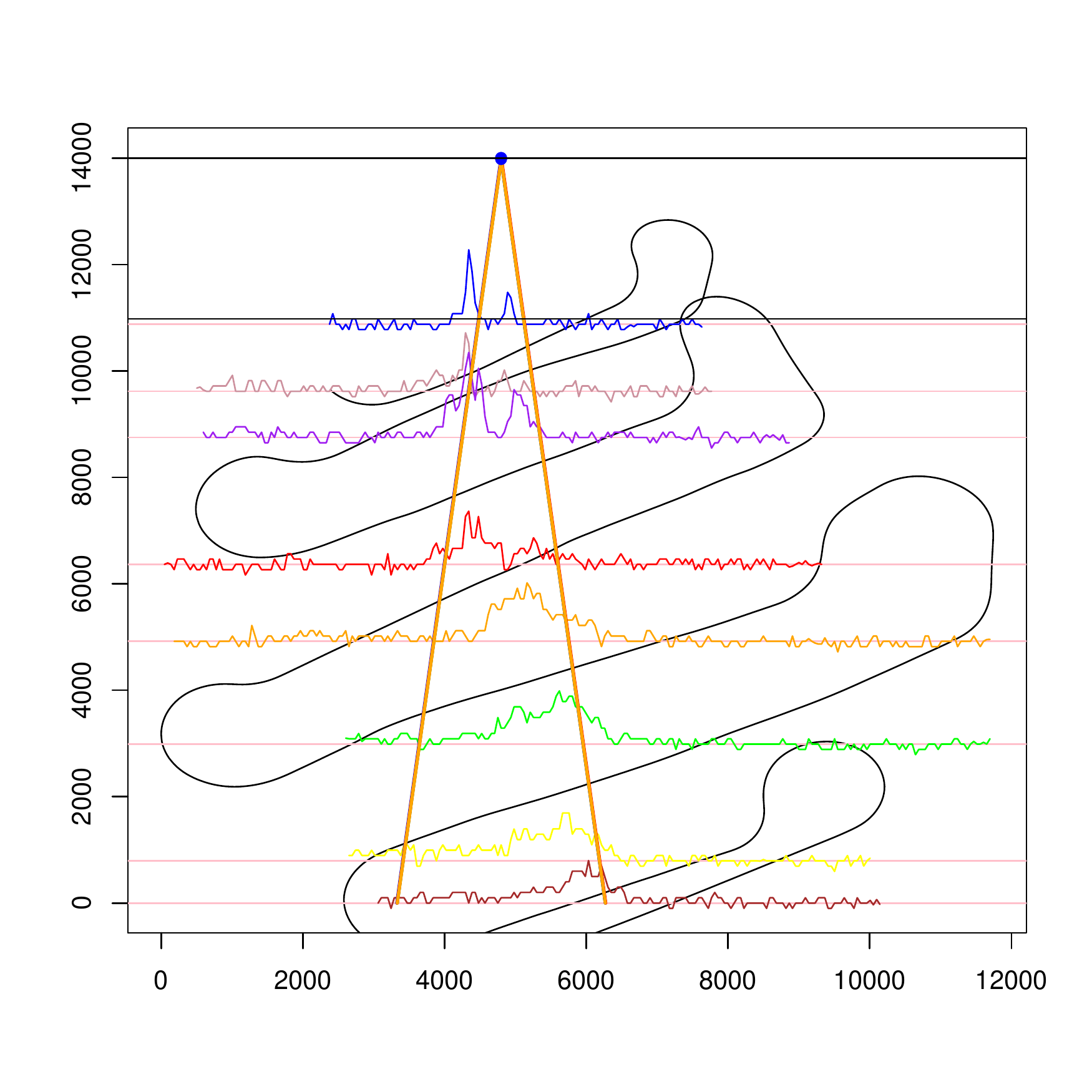}}
\end{subfloat}
\caption{Estimated {\cred clouds with unknown} opening angles for the right trajectory}
\label{right_traj_1}
\end{figure}

\begin{figure}[tb]
\begin{subfloat}[Estimated cloud based on the multivariate procedure]
	{\includegraphics[width=0.48\textwidth]{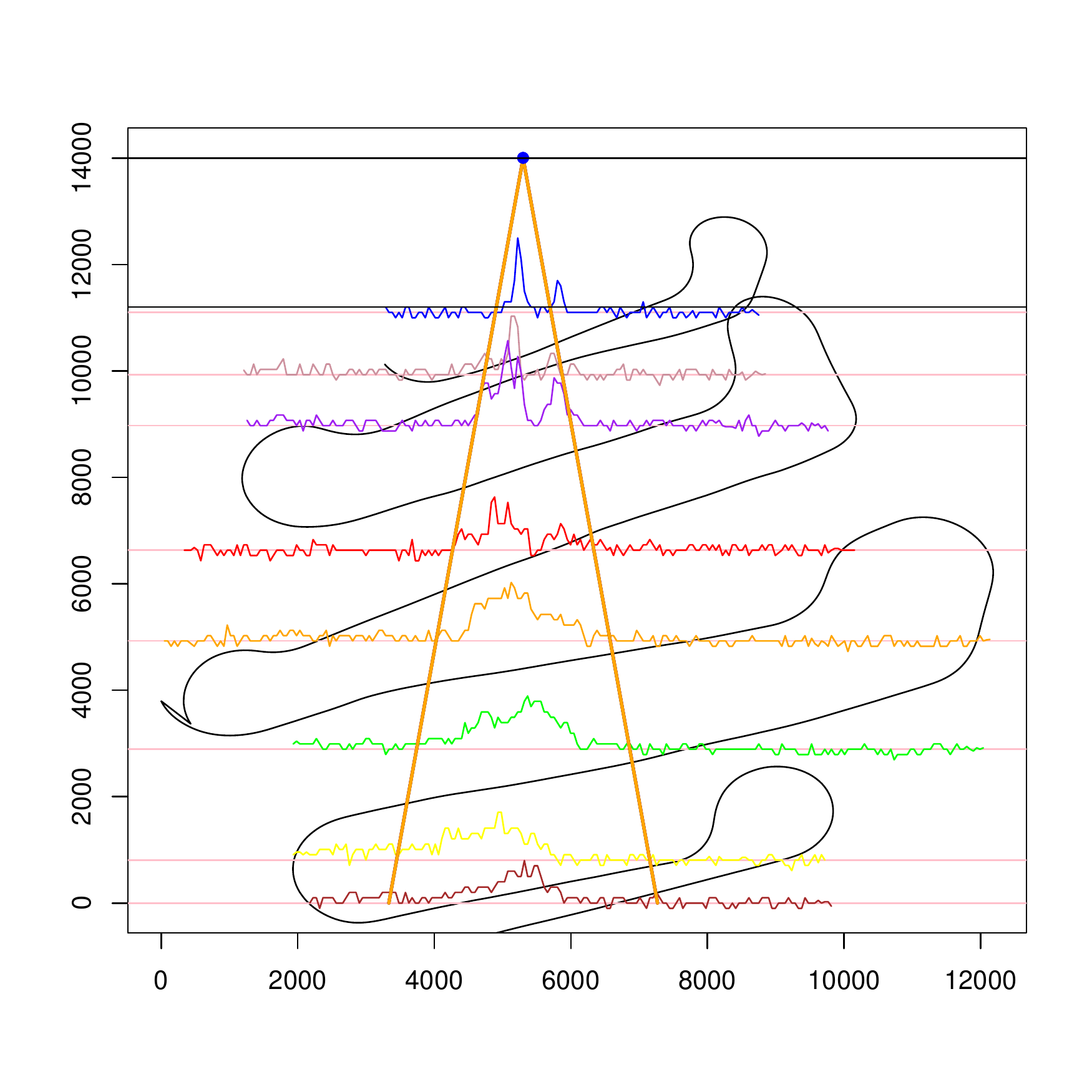}}
\end{subfloat}
\begin{subfloat}[Estimated cloud based on the projection procedure]
	{\includegraphics[width=0.48\textwidth]{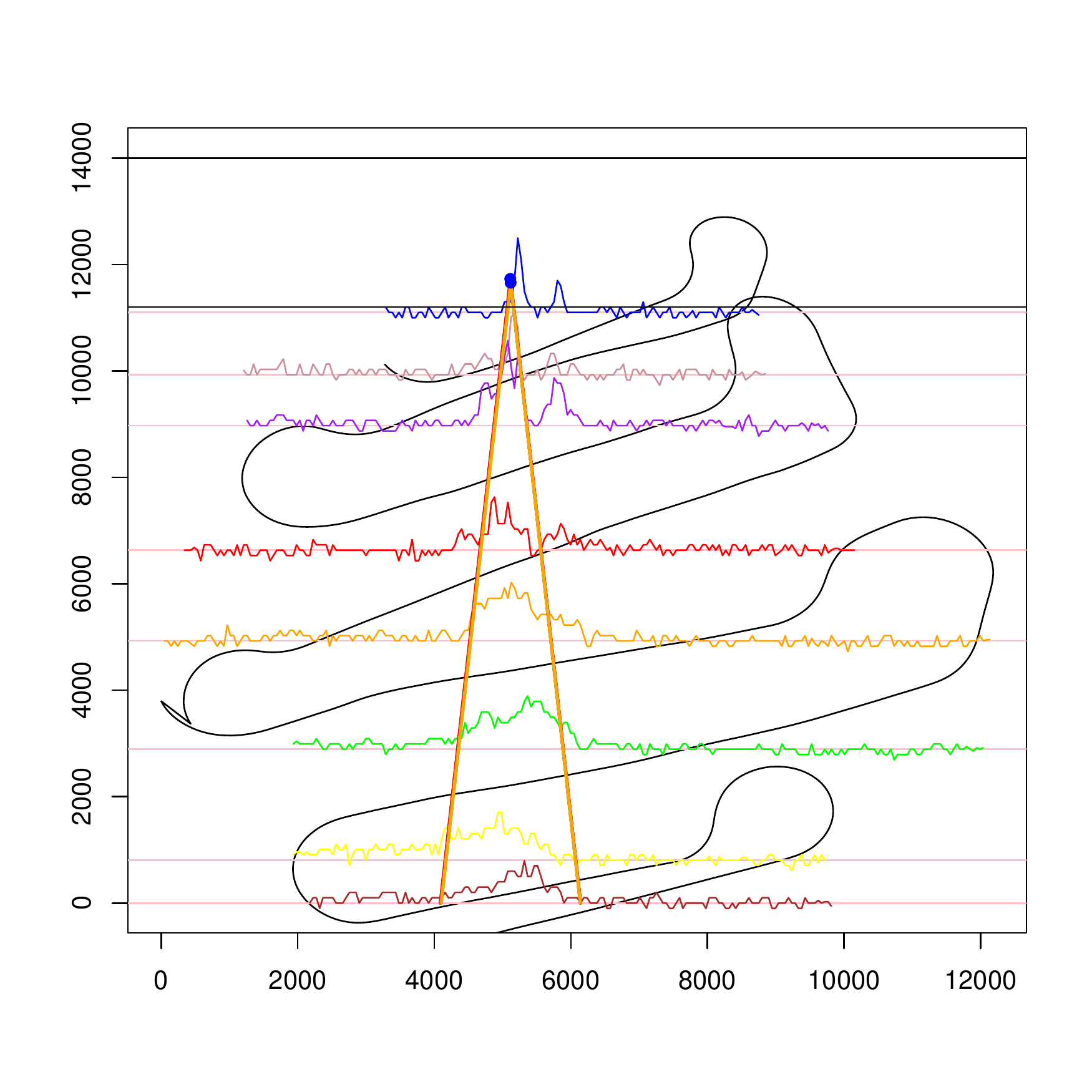}}
\end{subfloat}
\caption{Estimated opening angles {\cred with unknown opening angle} for the right trajectory after taking the wind change into account}
\label{right_traj_1b}
\end{figure}

For the analysis of the right trajectory we proceed analogously as for the left trajectory as described in Section~\ref{sec:data} noting that the results concerning the dependency structure are similar to the left trajectory (see Section 18 in~\cite{silke_diss}). The estimated clouds are given in Figure~\ref{right_traj_1}. A visual inspection shows that the cloud seems to have moved to the right in the lower legs {\cred such that the analysis of this data example with a linear cloud is an illustration how the method adapts to misspecification of the shape of the cloud}.  As already suggested by \cite{Hirst2013} this effect is likely to stem from a change in wind direction between the fourth and the fifth leg of the right flight path (some of the measurement locations are up to 15 km away).

Both of our estimation procedures (being bound to a linear cloud shape) compensate by choosing wider opening angles, while in the original Bayesian analysis of \cite{Hirst2013} their procedure compensated by suggesting a third spurious  source.  

{\cred Instead of using a more complicated cloud shape in the analysis, we take this } wind change into  account  by using different rotating angles in the preprocessing  which leads to a better aligned signal and consequently to more precise estimators (see Figure~\ref{right_traj_1b}). Looking at the heat maps in Figure~\ref{right_traj_2}, more pronounced for the projection statistic there seems to be two modes (i.e.\ areas with higher statistical values) which may still be an artefact of that wind change. Similarly to the identifiability issue along the $y$-axis this suggests that the heatmap is in fact a useful visualisation tool for this type of analysis for the location of the gas emission source based on aerial-sensed data.
\begin{figure}[tb]
\begin{subfloat}[Heatmap from the multivariate procedure]
{\includegraphics[width=0.98\textwidth]{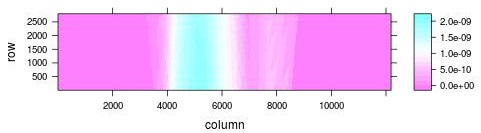}}
\end{subfloat}
\begin{subfloat}[Heatmap from the multivariate procedure]
{\includegraphics[width=0.98\textwidth]{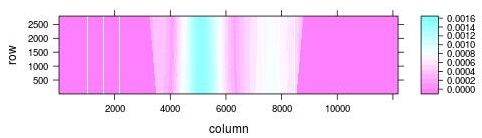}}
\end{subfloat}
\caption{Heatmaps for the right trajectory after taking the wind change into account}
\label{right_traj_2}
\end{figure}

 \end{document}